\documentclass[journal,comsoc]{IEEEtran}
\def\BibTeX{{\rm B\kern-.05em{\sc i\kern-.025em b}\kern-.08em
		T\kern-.1667em\lower.7ex\hbox{E}\kern-.125emX}}
\usepackage{epsfig,epstopdf,graphicx,psfrag,amsmath,cases}
\usepackage{latexsym,amssymb,amsmath,epsfig,algorithm,amsthm}
\usepackage{float}
\normalsize
\usepackage{cite}
\usepackage{graphicx}
\usepackage{amsmath}
\usepackage{relsize}
\usepackage{mathtools}
\DeclarePairedDelimiter{\abs}{\lvert}{\rvert}
\DeclarePairedDelimiter{\norm}{\lVert}{\rVert}
\usepackage{amsthm,amssymb}
\newtheorem{thm}{Theorem}
\newtheorem{cor}{Corollary}

\newtheorem{proposition}{Proposition}

\allowdisplaybreaks
\usepackage{longtable}
\usepackage{algorithmic}
\usepackage{array}
\usepackage{mdwmath}
\usepackage{mdwtab,bm}
\usepackage{cases}
\usepackage{eqparbox}
\usepackage[tight,footnotesize]{subfigure}
\usepackage{stfloats}
\newtheorem{definition}{Definition}
\newtheorem{example}{Example}
\theoremstyle{remark}

\usepackage{algorithm} 
\usepackage{algorithmic}

\usepackage{multirow}

\begin{document}
	\title{Design and Analysis of Delayed Bit-Interleaved Coded Modulation with LDPC Codes}
	\author{Yihuan~Liao,~\IEEEmembership{Student Member,~IEEE,}
			        Min~Qiu,~\IEEEmembership{Member,~IEEE,}
			        and~Jinhong~Yuan,~\IEEEmembership{Fellow,~IEEE}\thanks{This work was supported in part by the Australian Research Council Discovery Projects under Grant DP 190101363 and in part by the Linkage Project under Grant LP 170101196. This paper was presented in part at the 2019 IEEE Information Theory Workshop (ITW), Visby, Gotland, Sweden \cite{8989384}.} \thanks{Y. Liao, M. Qiu and J. Yuan are with the School of Electrical Engineering and Telecommunications, University of New South Wales, Sydney, NSW 2052, Australia (e-mail: yihuan.liao@student.unsw.edu.au; min.qiu@unsw.edu.au; j.yuan@unsw.edu.au).}}
	\markboth{IEEE Transactions on Communications}
	{Draft paper}

	\maketitle
	\begin{abstract}
		This paper investigates the design and performance of delayed bit-interleaved coded modulation (DBICM) with low-density parity-check (LDPC) codes. For Gray labeled square $M$-ary quadrature amplitude modulation (QAM) constellations, we investigate the optimal delay scheme with the largest spectrum efficiency of DBICM for a fixed maximum number of delayed time slots and a given signal-to-noise ratio. When analyzing the capacity of DBICM, we find two important properties: the capacity improvement due to delayed coded bits being mapped to the real and imaginary parts of the transmitted symbols are independent of each other; a pair of delay schemes with delayed coded bits having identical bit-channel capacity lead to equivalent DBICM capacity. Using these two properties, we efficiently optimize the delay scheme for any uniform Gray-QAM systems. Furthermore, these two properties enable efficient LDPC code designs regarding unequal error protection via bit-channel type classifications. Moreover, we use protograph-based extrinsic information transfer charts to jointly optimize degree distributions and channel assignments of LDPC codes and propose a constrained progressive edge growth like algorithm to jointly construct LDPC codes and bit-interleavers for DBICM, taking distinctive bit-channel's capacity into account. Simulation results demonstrate that the designed LDPC coded DBICM systems significantly outperform LDPC coded BICM systems.
	\end{abstract}
	\IEEEpeerreviewmaketitle

	\section{Introduction}
	
	Bit-interleaved coded modulation (BICM) \cite{bicm} is a pragmatic approach to achieve reliable communications with high spectrum efficiency by combining error correction codes and high order modulations. Uniform Gray labeled BICM has been extensively investigated for many wireless and optical communication systems \cite{bicm_wireless, szczecinski2015bit, 6844864, 4026714}, appreciating its close-to-coded modulation capacity performance. To achieve the best possible performance, BICM often employs modern channel codes, e.g., turbo codes \cite{bicm_turbo}, low-density parity-check (LDPC) codes \cite{Hou}, and polar codes \cite{bicm_polar}. In addition, quasi-cyclic LDPC codes designed for BICM schemes have been proposed in \cite{9050854} as a hardware friendly approach.
	
	Recently, a generalized BICM, namely, delayed-BICM (DBICM), was proposed in \cite{DBICM} to improve the transmission reliability over BICM. To be specific, DBICM modulates sub-blocks from multiple codewords and uses the decoded sub-blocks to improve the detection of other sub-blocks being modulated into the same signal sequence. Different from multilevel coding (MLC) that uses multiple channel codes with different code rates at each level \cite{580199,333853}, DBICM uses a single channel code. It is known that in BICM, each codeword is independently modulated and transmitted in a single time slot. However, in DBICM, a codeword is divided into the delayed and undelayed sub-blocks. Instead of transmitting all sub-blocks of each codeword within a single time slot, these sub-blocks are transmitted separately in multiple time slots. As a result, the decoding of a codeword is not performed until all sub-blocks of the codeword are received. Once decoded, the sub-blocks of the codeword can be viewed as known \cite{DBICM}. Knowing the delayed sub-blocks will effectively reduce the signal constellation size in demodulation. Therefore, the reliability of the demodulated signals for the undelayed sub-blocks is improved by using the extrinsic information of the decoded signals from delayed sub-blocks. As the reliability of the sub-blocks transmitted in the undelayed bit-channels improves, the capacities of these bit-channels are also increased. The authors in \cite{DBICM_cap} have developed tools to compute the capacity of DBICM for a given delay scheme, constellation, and labeling. It has been shown in \cite{DBICM, DBICM_cap} that the capacity of DBICM is bounded between the constellation constrained capacity and BICM capacity. The authors have also designed bit labels to achieve higher capacity improvement for DBICM under half $16$-QAM where Gray labeling does not apply. Another constellation labeling design for 16-QAM DBICM with iterative detection and decoding (DBICM-ID) has been investigated in \cite{8437452}. Despite their success, it is unclear how to design the delay schemes for a DBICM system to achieve the best possible spectrum efficiency for a fixed maximum number of delayed time slots. From our observation, in DBICM, randomly choosing a delay scheme may result in negligible performance gain over BICM. Moreover, a delay scheme that achieves a large DBICM capacity in some signal-to-noise ratio (SNR) region may not perform well in other SNR regions. To the best of our knowledge, the design of delay schemes and their impacts on the DBICM capacity have not been fully investigated in the literature.	 

	In addition to the delay scheme designs, it is also desirable to design good LDPC codes for DBICM systems. Extrinsic information transfer (EXIT) charts and progressive edge growth (PEG) \cite{PEG} are convenient tools for designing capacity-approaching LDPC codes with large girth in the Tanner graph. However, in high order modulated systems, such as BICM and DBICM, LDPC codes designed for uniform bit-channel capacity are not optimal due to the fact that the unequal error protection (UEP) among bit-channels is not considered. To address the UEP, previous works in \cite{1375222, bit_design_1, bit_design_2, bit_design_3, junyi_bit_mapping} have shown that it is capable of improving the bit error rate (BER) performance via bit mapper designs that map the positions of LDPC coded bits to appropriate BICM bit-channels according to their bit-channel capacities. By directly incorporating the different bit-channel properties in the code design \cite{Hou}, optimized irregular LDPC codes have been proposed in \cite{4024290} by adopting the extrinsic information transfer (EXIT) charts \cite{EXIT_chart}. Later, the authors in \cite{MET} extend the EXIT chart to multiple dimensions, from an edge perspective, to design multi-edge type (MET) LDPC codes for high order modulations. To design MET LDPC codes for modulation levels higher than $16$-QAM, they have suggested a high-order extension in a nesting fashion, which extends the edge types, starting from two, of their codes by optimizing only one additional edge type in each step. A simpler approach has been proposed in \cite{7339431} from a protograph-based approach, which is desirable to be applied for large constellations and obtains further decoding performance improvements over \cite{MET}. Specifically, \cite{7339431} represents each bit channel in BICM by a surrogate channel \cite{1638622} and jointly optimizes the protograph ensemble, also known as the base matrix, and the bit mapper for the surrogate channels. However, the limitation of the code design in \cite{7339431} lies in the integer nature of the base matrices. Recently, in \cite{junyi_regandirr}, the authors designed LDPC codes for uniform Gray-labeled 16-QAM which improves the decoding performance over \cite{MET,7339431} for the same constellation. They have proposed a variation of the MET-EXIT chart and jointly designed LDPC codes with bit mapping for BICM systems by dividing bit-channels into two types according to their capacities. However, the reliability of the edge can also be affected by the degree of VNs that are connected to the edge, which has not been considered by the MET-EXIT chart analysis for LDPC coded BICM schemes in \cite{junyi_regandirr}. For example, a VN assigned to an unreliable bit-channel with a high degree may be more reliable than a VN assigned to a reliable bit-channel with a low degree. Furthermore, there exist large capacity differences within both the reliable and unreliable channels in high order modulations, especially for DBICM systems. Without distinctive bit-channels with different capacities, the decoding threshold estimated by the MET-EXIT chart, as conducted in \cite{junyi_regandirr}, would be inaccurate in DBICM systems. 
	
	In this work, we focus on the design of LDPC coded DBICM schemes. We choose the protograph-based EXIT (PEXIT) chart \cite{PEXIT_chart} as an analytical tool to estimate the decoding threshold of the designed irregular LDPC codes, since the decoding thresholds of protograph LDPC codes in high order modulation systems computed from PEXIT chart and Monte-Carlo density evolution \cite{910577,910578} are close as shown in \cite{7339431}. However, instead of employing a base matrix in the code design, we directly construct the parity-check matrix of an LDPC code from its degree distributions. In addition, note that the iterative detection and decoding can improve the performance of BICM and DBICM at the cost of complexity and latency \cite{DBICM_cap,761047}. For some applications such as optical communications or video streaming services\cite{DBICM}, where latency is one of the main concerns, we will not consider iterative receivers for DBICM schemes in the paper. In the following, we summarize the main contributions of the work as below:
	\begin{itemize}
		\item We investigate the performance of uniform Gray labeled $M$-QAM DBICM systems. First, we analyzed the impacts of the delay scheme on DBICM capacity and developed a delay scheme optimization method. Specifically, we prove that the DBICM with Gray labeled square $M$-QAM constellations exhibits two properties: 1) Delaying any bit or a group of bits mapped to the real part of the modulated signals only affects the bit-channel capacity of the undelayed bits mapped to the real part of the signal constellations. It does not affect the bit-channel capacity of any delayed or undelayed bits mapped to the imaginary part of the signals, and vice versa. 2) Equivalent DBICM capacity can be achieved by a pair of delay schemes with delayed coded bits sharing identical bit-channel capacity. With these two properties, an $M$-QAM DBICM system can be decomposed into two independent but symmetric DBICM $\sqrt{M}$-PAM systems. Therefore, the problem of finding the optimal delay scheme for a uniform Gray labeled $M$-QAM DBICM system is reduced to finding the optimal delay scheme for the underlying uniform Gray labeled $\sqrt{M}$-PAM DBICM system. 
		\item By exploiting the properties and proposed method, we find the optimal delay scheme that allows DBICM to achieve a target spectrum efficiency under a fixed maximum number of delayed time slots $T_{\text{max}}$ with the lowest SNR. To be specific, we choose $T_{\text{max}} = 1$ to reduce the spectrum efficiency loss and find the optimal delay schemes for DBICM with uniform Gray labeled $16$-QAM, $64$-QAM, $256$-QAM, and $1024$-QAM. We show that the capacities of the $16$-QAM, $64$-QAM, $256$-QAM, and $1024$-QAM DBICM with optimized delay schemes are within $0$,  $0.1$, $0.3$ and $0.25$ dB away from their corresponding constellation constrained capacities, respectively.
		\item Based on the optimized delay scheme, we have designed BICM and DBICM LDPC codes jointly with bit mapping. By using the two properties of DBICM, we classify the bit-channels into several types, where the bit-channels within each type share identical bit-channel capacity. Furthermore, we introduce a channel assignment matrix $\mathbf{P}$ that assigns the variable node (VN) degree distribution $\lambda$ to each specific bit-channel type. Moreover, we propose a constrained PEG-like code construction algorithm to construct LDPC codes, according to channel assignments with large girth. Also, we propose a code optimization method and a constrained PEG-like code construction algorithm to jointly design LDPC code with bit mapping via optimizing channel assignment to obtain good belief propagation decoding thresholds, which are computed by the PEXIT chart.
		\item Numerical results show that the designed LDPC codes with the optimized delay schemes significantly outperform the LDPC coded BICM. More specifically, the designed $16$-QAM DBICM schemes outperform the corresponding BICM schemes by up to $0.5$ dB. For the $64$-QAM DBICM schemes, these improvements are up to $0.7$ dB. Furthermore, the performance gain obtained by the designed DBICM systems over their BICM counterparts increases with modulation level. 
	\end{itemize}
	
	Hereafter, we use normal font letters, such as $x$ and $X$, to represent scalars. Vectors and matrices are denoted by boldface letters such as $\mathbf{x}$ and $\mathbf{X}$. To show the elements in a length $N$ vector $\mathbf{x}$, we use $\mathbf{x}=[x_i]_{i=1}^N$ and $\mathbf{x}=[x_1,...,x_N]$ interchangeably in this paper. Furthermore, sets are represented by calligraphic letters, such as $\mathcal{X}$. Expectations are denoted by $\mathop{\mathbb{E}}[\cdot]$. $\mathbb{R}^{N\times M}$ denotes an $N \times M$ matrix with real entries. Moreover, we use $\mathbb{F}_2^{q}$ to represent a collection of binary vectors with length $q$. $\mathbf{0}_{a,b}$ and $\mathbf{1}_{a,b}$ represent a matrix of size $a \times b$ whose entries are either all zeros or ones, respectively. In addition, $\Re(\cdot)$ and $\Im(\cdot)$ represent the real and imaginary part of a symbol or set of symbols, respectively. 

	\section{System model}
	
	In this section, we describe the delayed bit-interleaved coded modulation (DBICM) systems over an additive white Gaussian noise (AWGN) channel. Particularly, Gray labeled uniform QAM is investigated in this paper. We focus on the design and analysis for DBICM without iterative detection and decoding schemes. Therefore, Gray labeled uniform QAM is selected as it is commonly used in BICM schemes without information feedback \cite{bicm, bicm_wireless}.
	
		\begin{figure}[h]
		\centering
		\includegraphics[width=0.45\textwidth]{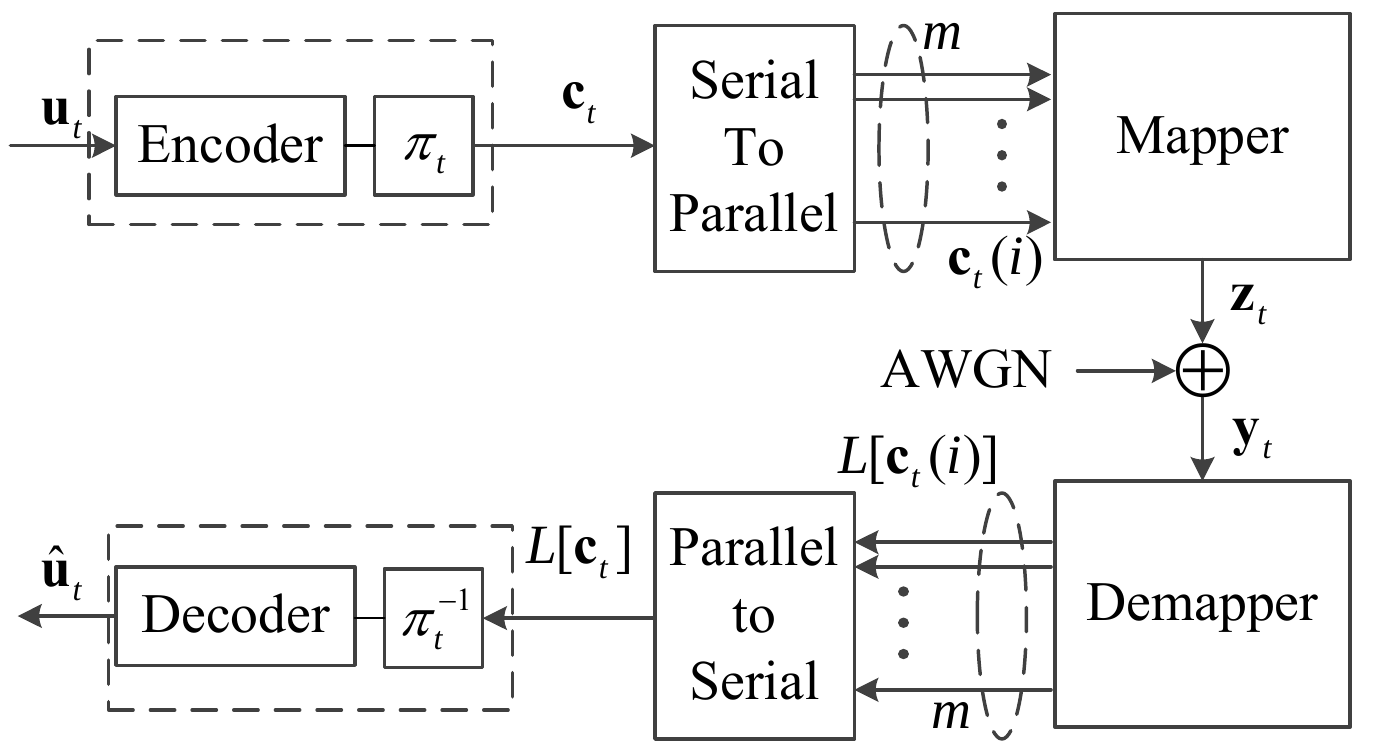}
		\caption{Block diagram of BICM transmitter and receiver structure.~~~~~~}
		\label{fig:BICM}
	\end{figure}
	
	\begin{figure}[h]
		\centering
		\includegraphics[width=0.45\textwidth]{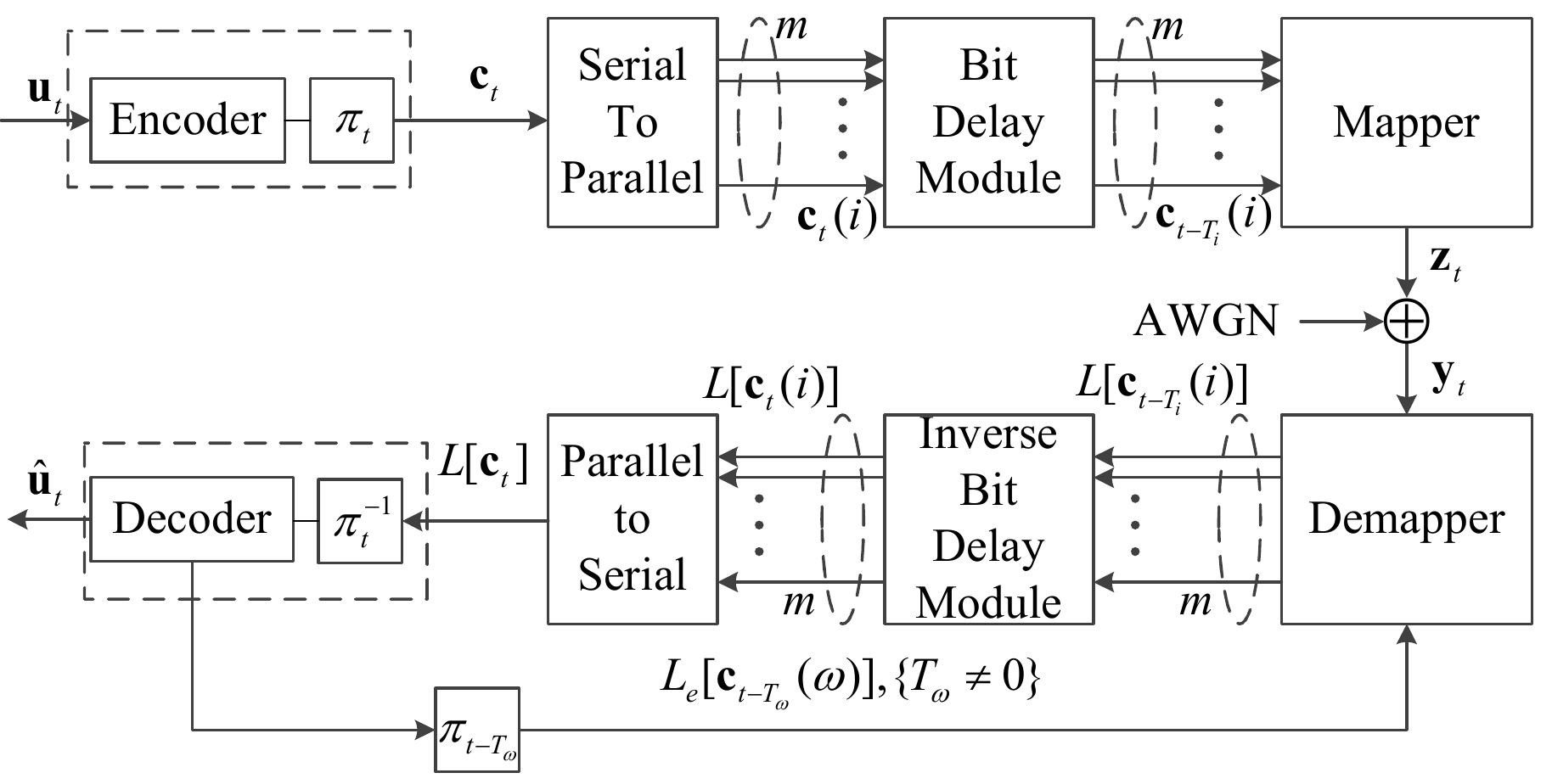}
		\caption{Block diagram of DBICM transmitter and receiver structure.~~~~~~}
		\label{fig:DBICM_structure}
	\end{figure}
	
	The block diagram of an LDPC coded BICM and DBICM system are shown in Fig. \ref{fig:BICM} and Fig. \ref{fig:DBICM_structure}, respectively. At the transmitter side, at time $t$, a binary information sequence $\mathbf {u}_{t}$ of length $K$ is encoded and interleaved to become a codeword $\mathbf{c}_{t}$ of length $N$ with code rate $R = K/N$. Here, the LDPC code and the interleaver are jointly designed. After a $1$-to-$m$ serial to parallel module, $m$ sub-blocks $\mathbf{c}_{t}(i), i\in\{0,...,m-1\}$, of $N/m$ bits are obtained, where the $m = \log_2{M}$ represents the modulation level. For simplicity, $N$ is assumed to be dividable by $m$. Next, each sub-block $\mathbf{c}_{t}(i)$ is delayed by $T_{i}$ time slots to generate $\mathbf{c}_{t-T_{i}}(i)$, $i\in\{0,...,m-1\}$. We use a delay scheme $\mathbf{T}=[T_i]_{i=0}^{m-1}$ to denote the number of delayed time slots for each sub-block which is known at both the transmitter and receiver side. Let $T_{\text{max}}$ and $T_{\text{min}}$ represent the maximum and minimum delayed time slot, respectively. We consider $T_{\text{min}}=0$. After the bit delay module, $m$ sub-blocks $\mathbf{c}_{t-T_i}(i), i\in\{0,1,...,m-1\}$, are mapped to a sequence of $M$-QAM symbols $\mathbf{z}_{t}$ of length $n = N/m$, where $\mathbf{z}_t = [z_t^0, z_t^1, ..., z_t^{n-1}]$. The modulated sequence of signals $\mathbf{z}_{t}$ is transmitted over an AWGN channel. We refer to each coded bit $b_i$ associated with the undelayed sub-blocks $\mathbf{c}_{t-T_i}(i)$ where $T_i = 0$ as the {\it{undelayed coded bit}}. Similarly, a coded bit associated with the sub-blocks $\mathbf{c}_{t-T_i}(i)$ where $T_i \neq 0$ is referred to as the {\it{delayed coded bit}}.
	
	The receiver obtains noisy signals $\mathbf{y}_{t} = \mathbf{z}_{t} + \mathbf{n}_{t}$, which are fed to the de-mapper, where $\mathbf{y}_t = [y_t^0, y_t^1, ..., y_t^{n-1}]$ and $\mathbf{n}_{t}$ represents the AWGN noise samples with zero mean and vairance $\sigma^2$ per real or imaginary dimension. The demapper performs a maximum a posteriori (MAP) symbol-to-bit metric calculation and its (soft) output can be derived as follows. Use $\chi$ to denote the Gray labeled $M$-QAM with a square signal constellation. For the $j$-th transmitted signal at time slot $t$, $z_t^j$, let $\mathbf{c}_t^j$ denotes its coded $m$ bit label $\mathbf{c}_t^j = [c_{t-T_0}^j(0),c_{t-T_1}^j(1),...,c_{t-T_{m-1}}^j(m-1)]$. For the received signal $y_t^j$, the demapper estimates the transmitted signal $\hat{z}_t^j$ and its corresponding bit label $\hat{\mathbf{c}}_t^j = [\hat{c}_{t-T_0}^j(0), \hat{c}_{t-T_1}^j(1), ..., \hat{c}_{t-T_{m-1}}^j(m-1)]$. Without any information feedback from the decoder, the demapper computes the log-likelihood ratios (LLRs) of each coded bit $c_{t-T_i}^j(i)$ from the $j$-th received signal $y_t^j$ in the received signal sequence $\mathbf{y}_t$, which is given by 
	\begin{equation} \label{eq:demo_without_ap}
	L[c_{t-T_i}^j(i)|y_t^j] = \ln\left(\dfrac{\sum\limits_{\{\hat{z}_t^j\in\chi|\hat{c}_{t-T_i}^j(i) = 0\}}e^{-\frac{\norm{{y}_t^j-{\hat{z}_t^j}}^2}{2\sigma^2}}}{\sum\limits_{\{\hat{z}_t^j\in\chi|\hat{c}_{t-T_i}^j(i) = 1\}}e^{-\frac{\norm{{y}_t^j-{\hat{z}_t^j}}^2}{2\sigma^2}}}\right).
	\end{equation}
	
	After demodulating all $m$ sub-blocks associated with the received signal sequence $\mathbf{y}_t$ , the obtained LLRs $L[\mathbf{c}_{t-T_i}(i)]$ are then passed to the inverse bit delay module to recover $L[\mathbf{c}_t(i)]$. Then, $L[\mathbf{c}_t(i)]$ enters an $m$-to-$1$ parallel to serial module.
	
	Note that at time $t$, the LLRs $\{L[c_{t-T_i}^j(i)|y_t^j]\}$ of all coded bits $\{c_{t-T_i}^j(i)\}$ for code $\mathbf{c}_{t-T_{\text{max}}}$ are obtained, these LLRs are sent to the decoder to recover the transmitted codeword $\mathbf{c}_{t-T_{\text{max}}}$.
	
	When the decoder successfully decodes the transmitted codeword $\mathbf{c}_{t-T_\text{max}}$ at time $t$, all the sub-blocks of the codeword $\mathbf{c}_{t-T_{\text{max}}}$ are known to the receiver. In this case, the decoder feeds back to the delayed sub-blocks to the demapper, which will refine its demodulation outputs for other sub-blocks at time slots $t'\in\{t-T_{\text{max}}+1, t-T_{\text{max}}+2, ..., t\}$. Let us assume that at time $t'$, the sub-blocks of $\mathbf{c}_{t'-T_{\omega}}(\omega)$ are known, where $T_{\omega}\neq 0$ representing the delay of the $\omega$-th sub-block for the codeword $\mathbf{c}_{t'-T_{\omega}}$. Then the undelayed sub-blocks and the delayed sub-blocks $\mathbf{c}_{t'-T_i(i)}$ with delay $T_i<T_{\omega}$ can be demodulated with the known $\mathbf{c}_{t'-T_{\omega}}(\omega)$. The LLR output of the demapper for these undelayed and less delayed coded bits is given by
	\begin{align} \label{eq:demo_with_hard}
	&L\left[c_{t'-T_i}^j(i)|y_{t'}^j, c_{t'-T_{\omega}}^j({\omega})\right] \notag \\
	= & \ln\left(\dfrac{\sum\limits_{\{\hat{z}_{t'}^j\in\chi|\hat{c}_{t'-T_i}^j(i)=0,\hat{c}_{t'-T_{\omega}}^j({\omega}) = c_{t'-T_{\omega}}^j({\omega})\}}e^{-\frac{\norm{{y}_{t'}^j-{\hat{z}_{t'}^j}}^2}{2\sigma^2}}}{\sum\limits_{\{\hat{z}_{t'}^j\in\chi|\hat{c}_{t'-T_i}^j(i)=1,\hat{c}_{t'-T_{\omega}}^j({\omega}) = c_{t'-T_{\omega}}^j({\omega})\}}e^{-\frac{\norm{{y}_{t'}^j-{\hat{z}_{t'}^j}}^2}{2\sigma^2}}}\right).
	\end{align} With the known sub-blocks of $\mathbf{c}_{t'-T_{\omega}}(\omega)$ at time $t'$, the demapper effectively has a reduced signal constellation size. Therefore, compared to the initial demapper output in Eq. (\ref{eq:demo_without_ap}), this refined demapper output will be improved. 
	
	On the other hand, when the decoding of the transmitted codeword $\mathbf{c}_{t-T_{\text{max}}}$ is unsuccessful at time $t$, the hard-decision of the codeword $\mathbf{c}_{t-T_{\text{max}}}$ is not reliable at the receiver. However, the estimated probability of a coded bit $c_{t-T_{\text{max}}}^j(i)$ being $b\in\{0,1\}$, $P_{b}(c^j_{t-T_{\text{max}}}(i))$, can be calculated from the decoder output. Let the LLR output for each coded bit from decoding the $j$-th transmitted signal at time slot $t$ be $L_e[c^j_{t-T_{\text{max}}}(i)]$. Then, we can compute the a priori information for each coded bit as
	\begin{equation} \label{eq:prob}
	\left\{ \begin{array}{ll}
	P_0(c^j_{t-T_{\text{max}}}(i)) = \frac{1}{1+e^{-L_e[c^j_{t-T_{\text{max}}}(i)]}},\\
	P_1(c^j_{t-T_{\text{max}}}(i)) = 1-\frac{1}{1+e^{-L_e[c^j_{t-T_{\text{max}}}(i)]}}.
	\end{array}\right.
	\end{equation} In this case, the soft-decision feedback from the decoder can also refine the demodulation outputs for \textit{other sub-blocks} at time slots $t'\in\{t-T_{\text{max}}+1, t-T_{\text{max}}+2, ..., t\}$, which has not been decoded yet. Note that the soft-decision feedback in DBICM is different from iterative detection and decoding which passes the LLR and extrinsic information from the demapper and the decoder to the current sub-block at each decoding iteration \cite{1237404}, as our scheme only passes the LLR and extrinsic information from the demapper and the decoder to \textit{other sub-blocks once} for each sub-block. At time $t'$, the estimated a priori probability of each coded bit in the sub-blocks of $\mathbf{c}_{t'-T_{\omega}}(\omega)$ is computed following Eq. (\ref{eq:prob}), where $T_{\omega}\neq 0$ representing the delay of the $\omega$-th sub-block for the codeword $\mathbf{c}_{t'-T_{\omega}}$. Then the undelayed sub-blocks and the delayed sub-blocks $\mathbf{c}_{t'-T_i(i)}$ with delay $T_i<T_{\omega}$ can be demodulated with the estimated probability of each bit in $\mathbf{c}_{t'-T_{\omega}}(\omega)$. The LLR output of the demapper for these undelayed and less delayed coded bits is given by 
	\begin{align} \label{eq:demo_with_ap}
	& L[c_{t'-T_i}^j(i)|y_{t'}^j,L_e[c_{t'-T_{\omega}}^j({\omega})]] \notag \\
	= & \ln\left(\dfrac{\sum\limits_{b=0}^{1}\sum\limits_{\{\hat{z}_{t'}^j\in\chi|\hat{c}_{t'-T_i}^j=0,\hat{c}_{t'-T_{\omega}}^j({\omega}) = b\}}e^{-\frac{\norm{{y}_{t'}^j-{\hat{z}_{t'}^j}}^2}{2\sigma^2}}P_{b}(c^j_{t'-T_{\omega}}({\omega}))}{\sum\limits_{b=0}^{1}\sum\limits_{\{\hat{z}_{t'}^j\in\chi|\hat{c}_{t'-T_i}^j=1,\hat{c}_{t'-T_{\omega}}^j({\omega})=b\}}e^{-\frac{\norm{{y}_{t'}^j-{\hat{z}_{t'}^j}}^2}{2\sigma^2}}P_{b}(c^j_{t'-T_{\omega}}({\omega}))}\right).
	\end{align}At time $t'$, with the estimated probability for sub-blocks of $\mathbf{c}_{t'-T_{\omega}}(\omega)$, the demapper effectively has a probability weighted signal constellation. Therefore, compared to the initial demapper output in Eq. (\ref{eq:demo_without_ap}), this refined demapper output will also be improved.
	
	Although DBICM is capable of improving the reliability in the demodulation, it inevitably creates a loss in spectrum efficiency due the vacant sub-blocks resulted from the delay. Use $T_n$ to represent the total number of time slots in a transmission frame in BICM, and the total number of time slots in a frame in DBICM is $T_n + T_{\text{max}}$. Then, the spectrum efficiency of the DBICM is 
	\begin{equation}\label{eq:spectrum_eff}
	\eta_{\text{DBICM}} = mR\left(\frac{T_n}{T_n + T_{\text{max}}}\right),
	\end{equation}
	where $mR$ is the spectrum efficiency of BICM. When the number of time slots $T_n$ in a transmission frame is large, the spectrum efficiency loss can be seen as negligible. In addition, following Eq. (\ref{eq:spectrum_eff}), to minimize the spectrum efficiency loss, we only consider that the sub-blocks of the codewords are delayed by no more than one time slot, i.e., $T_{\text{max}} = 1$, throughout the paper. But the analysis and designs presented in this paper can be easily generalized to other delay schemes. We emphasize that DBICM decoder outputs decoded codewords in a streaming fashion, allowing users to receive and decode files continuously. Now with $T_{\text{max}} = 1$, the decoder is able to receive and output a codeword for each time slot (by using the undelayed sub-blocks from the previous time slot). Hence, DBICM is suitable for video streaming applications as the increased decoding latency is very minor for $T_{\text{max}}=1$ and the loss of spectrum efficiency is negligible.
	
	\begin{figure}[h]
		\centering
		\includegraphics[width=0.45\textwidth]{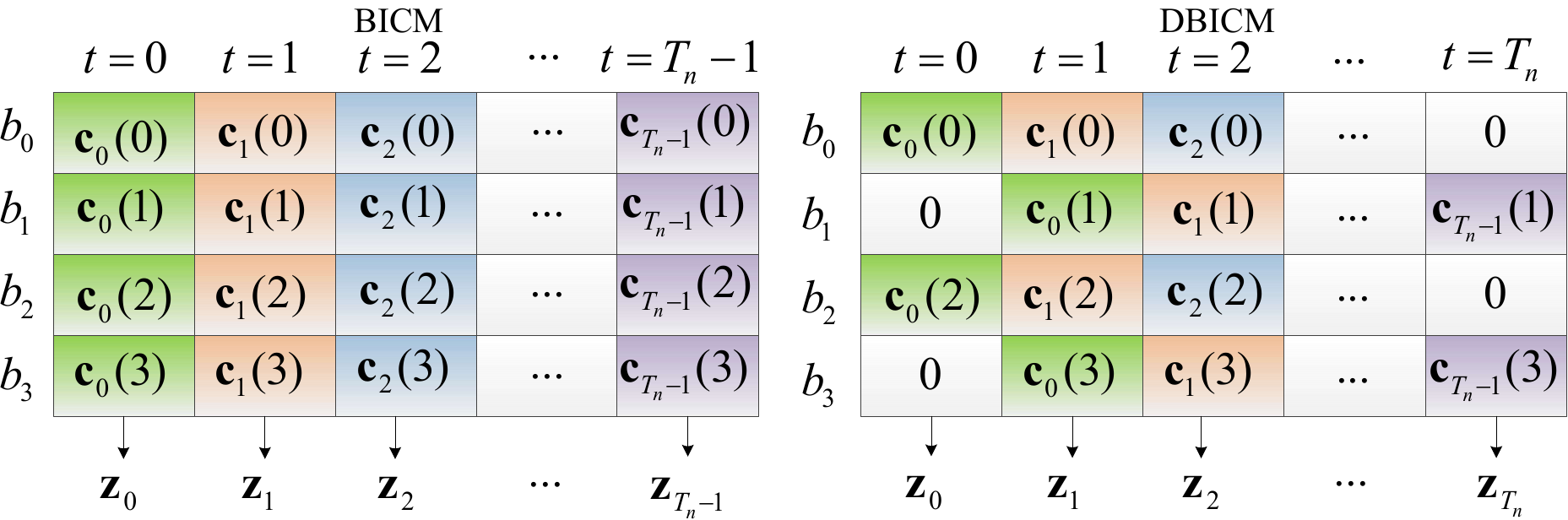}
		\caption{An example of $16$-QAM BICM and DBICM with a delay scheme $\mathbf{T}=[0,1,0,1]$, representing sub-blocks grouping for delays $T_0=0$, $T_1 = 1$, $T_2 = 0$, $T_3 = 1$.~~~~~~}
		\label{fig:delay_example}
	\end{figure}
	
	It is worth mentioning that the decoder of DBICM only decodes each codeword once, which is the same as that of BICM. Therefore, the decoding complexity of DBICM is equivalent as that of BICM. On the other hand, DBICM has higher detection complexity than BICM due to multiple detections on the undelayed sub-blocks. However, since we consider DBICM with $T_{\text{max}}=1$, only the undelayed sub-blocks are detected by one more time as compared with BICM. Therefore, the detection complexity for DBICM systems investigated in this paper is manageable.
	
	In Fig. \ref{fig:delay_example}, we use an example of $16$-QAM BICM and DBICM to demonstrate the difference between BICM and DBICM schemes. Let the $4$ bit-channels for $16$-QAM be represented by $b_0$, $b_1$, $b_2$ and $b_3$. In the BICM scheme, each codeword $\mathbf{c}_{t}$ is modulated to a sequence of signals $\mathbf{z}_t$ and transmitted at time slot $t$. On the other hand, for DBICM, sub-blocks $\mathbf{c}_{{t}-{T}_{i}}(i)$, $i\in \{0,...,m-1\}$, from multiple codewords are modulated to a sequence of signals $\mathbf{z}_t$ and transmitted at time slot $t$. Here, the DBICM system with delay scheme $\mathbf{T}=[0,1,0,1]$, indicating sub-blocks $\mathbf{c}_{t}(1)$ and $\mathbf{c}_{t}(3)$ are delayed by one time slot. Note that, the DBICM systems need $T_{\text{max}}$ more time slots to transmit the same number of codewords as that of BICM systems. Here, in Fig. \ref{fig:delay_example}, the DBICM scheme need $T_n + 1$ time slots to transmit $T_n$ codewords. Furthermore, in DBICM, each codeword $\mathbf{c}_{t}$ is divided into the delayed sub-blocks, $\mathbf{c}_t(1)$ with $\mathbf{c}_t(3)$, and the undelayed sub-blocks $\mathbf{c}_t(0)$ with $\mathbf{c}_t(2)$. Consequently, the delay also creates vacant places at time slot $t=0$ and $t=T_n$. We fill these vacant places as all-zero, and they are known by both the transmitter and the receiver sides. For some applications such as optical communications and video streaming, where the transmission frame $T_n$ is usually large compared to $T_{\text{max}}=1$, the DBICM systems investigated in this paper have a negligible loss of spectrum efficiency over BICM according to Eq. (\ref{eq:spectrum_eff}).
	
	It has been shown in \cite{DBICM_cap} that iterative detection and decoding is possible to further improve the demapping and decoding performance for DBICM at the price of latency and complexity. However, for practical consideration, this paper focuses on DBICM without iterative processing.		
	
	\section{Capacity analysis for DBICM}
	
	In this section, we discuss the capacity of the BICM and DBICM schemes. We also provide essential definitions that will be useful in the later analysis. Let us consider a system with a memory-less discrete-input and continuous-output channel, with input $\mathbf{z}$ and output $\mathbf{y}$. For simplicity, we drop the subscript and superscript of symbols $\mathbf{z}$ and $\mathbf{y}$ with slightly abused notations. 
	
	For a BICM scheme with the ideal interleaving assumption, all bit-channels are independent \cite{bicm}. In this case, the $i$-th bit-channel capacity, note as $C_{i,\text{BICM}}$, can be written as the mutual information between the input $b$ and its channel output $\mathbf{y}$, as shown in \cite{bicm} 
	\begin{align}
	\label{eq:C_bicm_bit}
	\qquad C_{i,\text{BICM}} = 1-\mathop{\mathbb{E}}_{b,\mathbf{y}} \left [\log_2\dfrac{\sum_{\mathbf{z} \in \chi}p(\mathbf{y}|\mathbf{z})}{\sum_{\mathbf{z} \in \chi^{i}_b}p(\mathbf{y}|\mathbf{z})} \right],
	\end{align} 
	where $\chi^{i}_b$ denotes the subset of all the signals $\mathbf{z}\in\chi$ whose $i$-th bit being $b\in\{0,1\}$, and \begin{equation} \label{eq:p}
	p(\mathbf{y}|\mathbf{z}) = e^{-\frac{\norm{\mathbf{y}-\mathbf{z}}^2}{2\sigma^2}}
	\end{equation} 
	is the channel conditional probability density function, while $\sigma^2$ being the noise variance per real dimension. Then, the capacity of the BICM scheme can be computed as \cite{bicm}
	\begin{equation}
	\label{eq:C_bicm_overall}
	C_{\text{BICM}} = \sum\limits_{i=0}^{m-1}C_{i,\text{BICM}} = m-\sum\limits_{i=0}^{m-1}\mathop{\mathbb{E}}_{b,\mathbf{y}} \left [\log_2\dfrac{\sum_{\mathbf{z} \in \chi}p(\mathbf{y}|\mathbf{z})}{\sum_{\mathbf{z} \in \chi^{i}_b}p(\mathbf{y}|\mathbf{z})} \right].
	\end{equation} 
	Now, we introduce the following definition related to the symmetric property of Gray labeled square $M$-QAM BICM systems.
	\begin{definition}\label{def:mcbpp}
		A pair of bits, $j$ and $j' \in \{1,...,m\}$, and $j \neq j'$, in the Gray labeled $M$-QAM are said to be symmetric bits if they have identical bit-channel capacities, such that
		\begin{equation} \label{eq:bj_bicm}
		C_{j,\text{BICM}}=C_{{j'},\text{BICM}}.
		\end{equation}
	\end{definition} 
	For simplicity, we consider decomposable Gray labeled $M$-QAM with square signal constellations which are the Cartesian product of the real and complex Gray labeled $\sqrt{M}$-PAM constellations. To be specific, in the Gray labeled $M$-QAM, the first $m/2$ bits with labels $\mathcal{A} = \{0,...,\frac{m}{2} - 1\}$ are mapped to the real part of the constellation, while the rest of bits with labels $\mathcal{B} = \{\frac{m}{2}, ..., m - 1\}$ are mapped to the imaginary part of the constellation. As a result, the BICM scheme has symmetric bit-channel capacities between set $\mathcal{A}$ and $\mathcal{B}$, since the real and imaginary parts are independent and symmetric. Recall Definition \ref{def:mcbpp}, each labeled bit in $\mathcal{A}$ has a symmetric bit in $\mathcal{B}$, and vice versa. 
	
	For a DBICM scheme, the bit-channel capacity of an undelayed coded bit is conditioned on the decoded information of the delayed sub-blocks under the delay scheme $\mathbf{T}$. Let us denote by $\mathcal{D}=\{i|T_i\neq 0\}$ and $\mathcal{\tilde{D}}=\{i|T_i = 0\}$ the collection of coded bit labels for the corresponding delayed and undelayed sub-blocks $\mathbf{c}_{t-T_i}(i)$, respectively. Let $\mathbf{b}_{\mathcal{D}}\in\mathbb{F}_{2}^{\abs{\mathcal{D}}}$ be a realization of the delayed coded bits, i.e., each delayed coded bit has a value of either 0 or 1. Then, under the delay scheme $\mathbf{T}$, the DBICM bit-channel capacity of an undelayed coded bit with label $k\in\mathcal{\tilde{D}}$, $C_{k,\text{DBICM}}^{\mathbf{T}}$, is given by \cite{DBICM_cap} 
	\begin{align}
		\label{eq:C_dbicm_bit}
		&C_{k,\text{DBICM}}^{\mathbf{T}} = I(b;\mathbf{y}|\mathcal{D}) \notag \\
		= & 1 - \frac{1}{2^{\abs{\mathcal{D}}}}\mathlarger{\sum_{\mathbf{b}_{\mathcal{D}}\in\mathbb{F}_{2}^{\abs{\mathcal{D}}}}}\mathop{\mathbb{E}}_{b,\mathbf{y}|\mathcal{D}} \left [\!\log_2\dfrac{\sum_{\mathbf{z} \in \chi^{\mathcal{D}}_{{\mathbf{b}_{\mathcal{D}}}}}p(\mathbf{y}|\mathbf{z})}{\sum_{\mathbf{z} \in \chi^{k,\mathcal{D}}_{b,{\mathbf{b}_{\mathcal{D}}}}}p(\mathbf{y}|\mathbf{z})}\! \right],
	\end{align} 
	where $\chi^{\mathcal{D}}_{{\mathbf{b}_{\mathcal{D}}}}\subseteq \chi$ represents the set of constellation points where the label of the delayed coded bits being $\mathbf{b}_{\mathcal{D}}$, $\chi^{k,\mathcal{D}}_{b,{\mathbf{b}_{\mathcal{D}}}} \subseteq \chi^{\mathcal{D}}_{{\mathbf{b}_{\mathcal{D}}}}$ is the collection of constellation points where the label of the delayed coded bits being $\mathbf{b}_{\mathcal{D}}$ and the label of the $k$-th undelayed coded bit being $b$. To show the relationship among $\chi$, $\chi^{\mathcal{D}}_{{\mathbf{b}_{\mathcal{D}}}}$ and $\chi^{k,\mathcal{D}}_{b,{\mathbf{b}_{\mathcal{D}}}}$, we give an example as follows.
	
	\begin{example}
		Consider a Gray labeled $64$-QAM DBICM system with  $\mathbf{T}=[0,0,0,0,1,1]$ corresponds to $\mathcal{D}=\{4,5\}$. There are $2^{\abs{\mathcal{D}}}=4$ different $\mathbf{b}_{\mathcal{D}}$, i.e., $[0,0], [0,1], [1,0]$, and $[1,1]$, and each $\mathbf{b}_{\mathcal{D}}$ creates partitions in the constellation. We use four different markers, as indicated in Fig. \ref{fig:gray_64_const_sub_group_a}, to represent the constellation points of $\chi^{\mathcal{D}}_{{\mathbf{b}_{\mathcal{D}}}}$ under four different realizations of $\mathbf{b}_{\mathcal{D}}$.
		
		To investigate the term inside the expectation of Eq. (\ref{eq:C_dbicm_bit}) under the delay scheme $\mathbf{T}$, it is necessary to look into the subset  $\chi^{k,\mathcal{D}}_{b,{\mathbf{b}_{\mathcal{D}}}}$ regarding the undelayed bit with label $k$ being $b$ and the delayed coded bits $\mathcal{D}$ being ${\mathbf{b}_{\mathcal{D}}}$. Take the third undelayed bit ($k=2$) as an example. In Fig. \ref{fig:gray_64_const_sub_group_b}, we use red and blue colors to represent the constellation points where the label of the third coded bit being $0$ and $1$, respectively. Thus, different colors and markers visualize $8$ subsets $\chi^{k,\mathcal{D}}_{b,{\mathbf{b}_{\mathcal{D}}}}$, as shown in Fig. \ref{fig:gray_64_const_sub_group_b}. 
		\begin{figure}[h]
			\centering
			\subfigure[${\mathbf{b}_{\mathcal{D}}}$ partition $\chi$ into subsets $\chi^{\mathcal{D}}_{{\mathbf{b}_{\mathcal{D}}}}$.]
			{
				\includegraphics[width=0.45\textwidth]{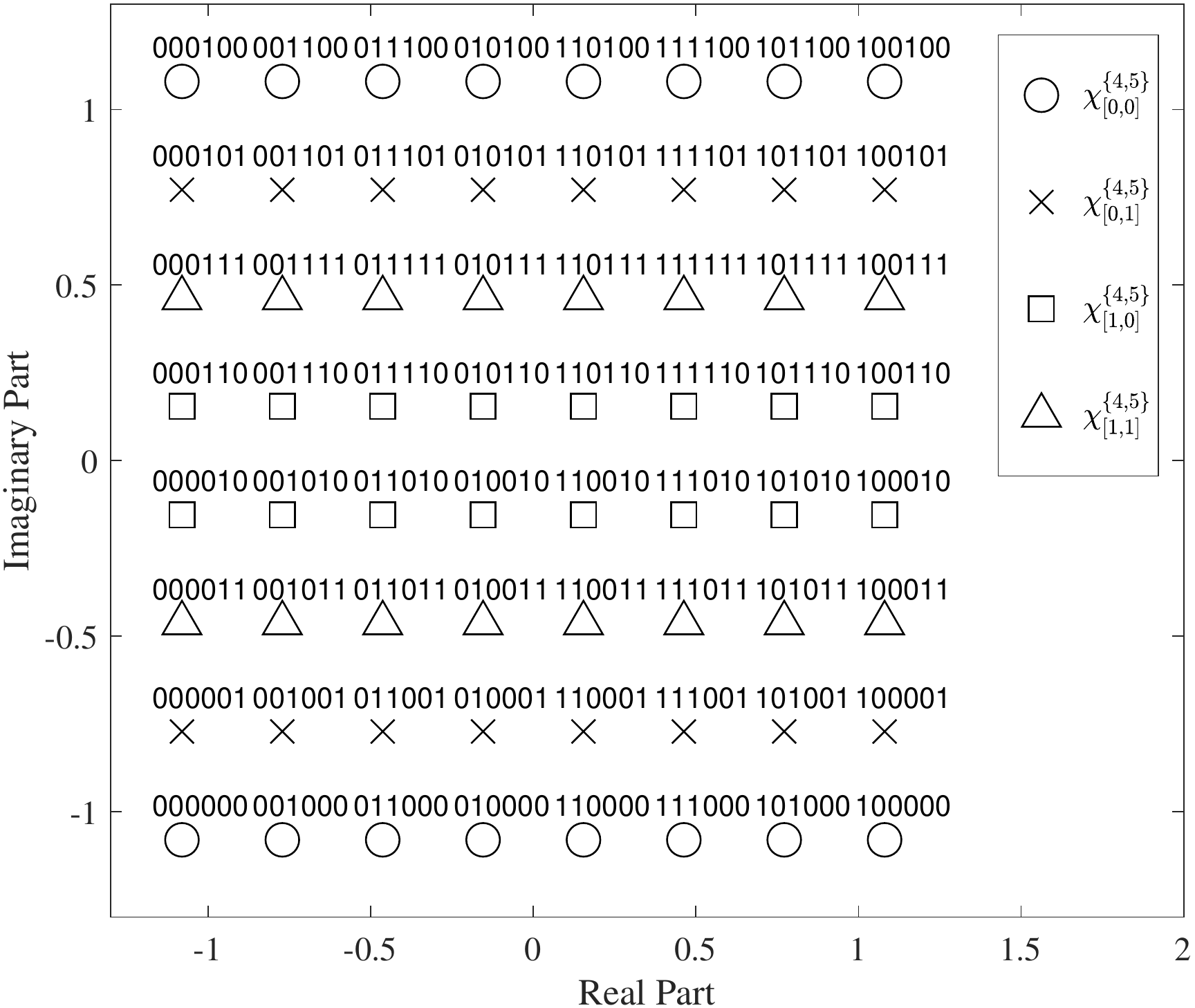}
				\label{fig:gray_64_const_sub_group_a}
			}
			\subfigure[$b$ and ${\mathbf{b}_{\mathcal{D}}}$ partition $\chi$ into subsets $\chi^{k,\mathcal{D}}_{b,{\mathbf{b}_{\mathcal{D}}}}$.]
			{
				\includegraphics[width=0.45\textwidth]{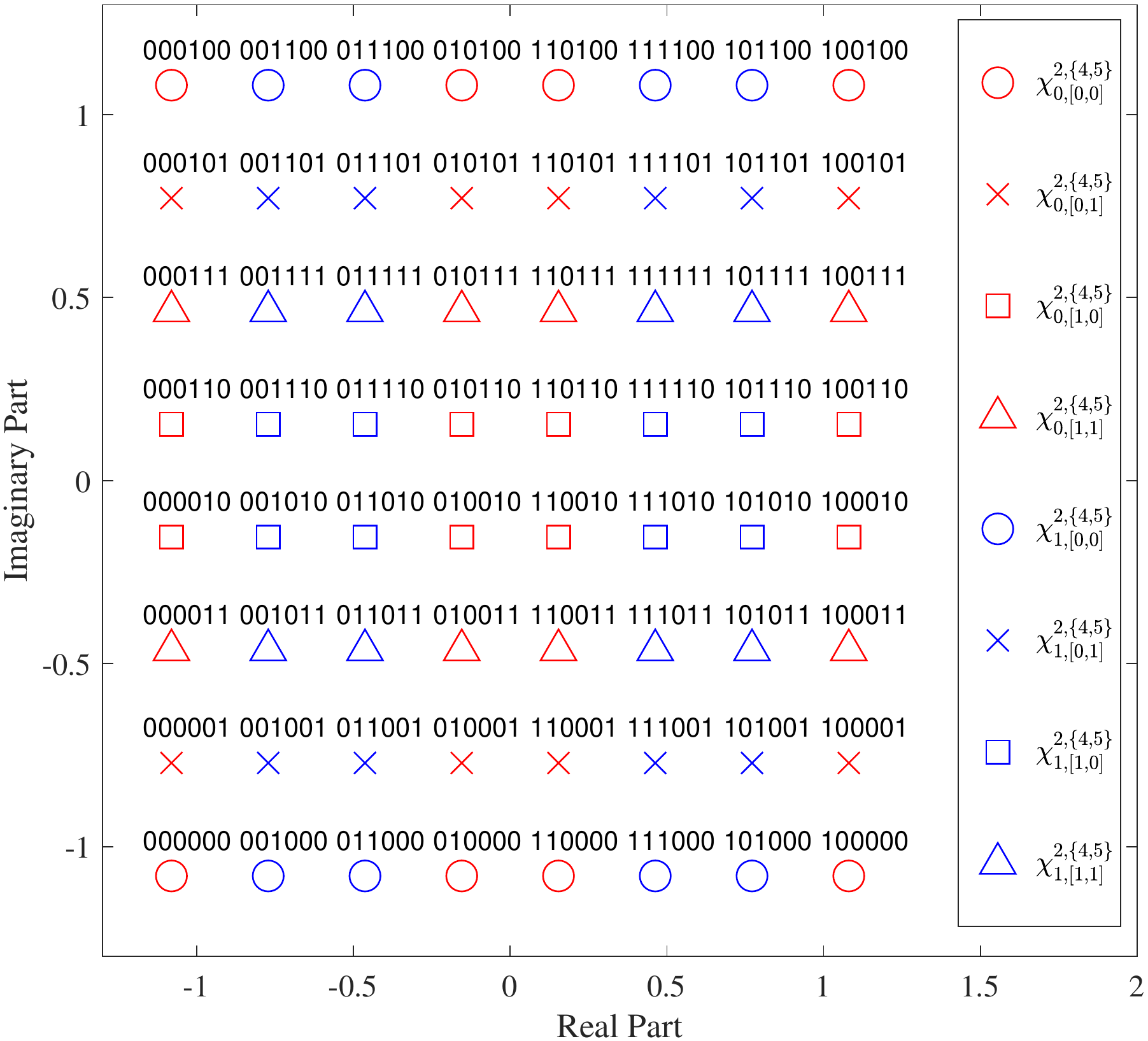}
				\label{fig:gray_64_const_sub_group_b}
			}
			\caption{An example of DBICM constellation partition.}
			\label{fig:sample_subfigures}
		\end{figure}
	\end{example} 
	
	For the delayed coded bits, their bit-channel capacities are unchanged from BICM. In this way, the DBICM capacity for the delay scheme $\mathbf{T}$, note as $C_{\text{DBICM}}^{\mathbf{T}}$, can be represented as the sum of the delayed and undelayed coded bits' capacities, calculated separately from Eq. (\ref{eq:C_bicm_bit}) and Eq. (\ref{eq:C_dbicm_bit}), which is given by 
	\begin{align}
	\label{eq:C_dbicm_overall}
	C_{\text{DBICM}}^{\mathbf{T}} = \sum\limits_{j\in\mathcal{D}}C_{j,\text{BICM}} + \sum\limits_{k\in\tilde{\mathcal{D}}}C_{k,\text{DBICM}}^{\mathbf{T}},
	\end{align} 
	where $j$ and $k$ label the delayed and undelayed coded bits, separately. Since the DBICM capacity depends on delay schemes, the design of delay schemes is an important issue in DBICM, which will be discussed in the next section.
	
	\section{Design of DBICM delay schemes} 
	In this section, we first show several characteristics of DBICM capacity over AWGN channels. We show the comparison of spectral efficiency of Gray labeled $M$-QAM DBICM systems with different delay schemes and their underlying $\sqrt{M}$-PAM DBICM systems. Then, we determine the optimal delay scheme under a fixed $T_{\text{max}}$ that allows Gray labeled $M$-QAM DBICM to achieve a target spectrum efficiency with the lowest SNR. We emphasize that randomly choosing the delay scheme may result in negligible improvement from BICM capacity. Therefore, finding the optimal delay scheme is crucial for obtaining the maximum capacity gain of DBICM over BICM. 
	
	First, we note that when $T_{\text{max}}=m-1$, where $m=\log_2M$ is the modulation level, the delay scheme $\mathbf{T}=[0,1,\cdots,m-1]$ and its permutation, achieve the constellation constrained capacity, regardless of constellations, labeling, and SNR, as shown in Appendix \ref{app:added_proof}. However, to construct such a capacity-achieving delay scheme, the maximum number of delayed time slots $T_{\text{max}}$ needs to increase with the constellation size. As a result, a large $T_{\text{max}}$ could introduce high decoding latency and complexity as the decoding of each codeword requires to detect the signals received during $T_{\text{max}}+1$ time slots. Since we restrict ourselves with $T_{\text{max}}=1$ to minimize the spectrum efficiency loss, we need to find the optimal delay scheme under this condition that achieves a target spectrum efficiency with the lowest SNR. 

	\subsection{Properties of the Delay Schemes}
	In this section, we present two important properties and a proposition for the capacity of the DBICM schemes. First, the bit-channel capacities of a Gray labeled $M$-QAM DBICM system can be independently considered via two $\sqrt{M}$-PAM DBICM systems. We now formulate this independent property as the following theorem.

	\begin{thm} \label{thm:superposition}
		Let the delay scheme $\mathbf{T} = [T_i]_{i=0}^{m-1} = [\mathbf{T}_{\mathcal{A}}, \mathbf{T}_{\mathcal{B}}]$, $\mathbf{T}_{\mathcal{A}} = [T_{i}]_{i \in \mathcal{A}}$, $\mathbf{T}_{\mathcal{B}} =  [T_{i}]_{i \in \mathcal{B}}$, $\mathcal{A}=\{0,1,...,\frac{m}{2}-1\}$ and $\mathcal{B}=\{\frac{m}{2}, \frac{m}{2} + 1,...,m-1\}$. For any $\mathbf{T}_{\mathcal{A}}$, $\mathbf{T}_{\mathcal{B}}$, and $\mathbf{T'}_{\mathcal{B}}$, where $\mathbf{T}_{\mathcal{B}}\neq\mathbf{T'}_{\mathcal{B}}$, we have the $k$-th undelayed bit-channel capacity, when $k\in \mathcal{A}$, as
		\begin{align}\label{eq:thm_1_1}
		&C^{[\mathbf{T}_{\mathcal{A}},\mathbf{T}_{\mathcal{B}}]}_{k,\text{DBICM}}(M\text{-QAM})=C^{[\mathbf{T}_{\mathcal{A}},\mathbf{T'}_{\mathcal{B}}]}_{k,\text{DBICM}}(M\text{-QAM}) \notag \\
		= & C^{\mathbf{T}_{\mathcal{A}}}_{k,\text{DBICM}}(\sqrt{M}\text{-PAM}).
		\end{align}Alternatively, for any $\mathbf{T}_{\mathcal{A}}$, $\mathbf{T'}_{\mathcal{A}}$, and $\mathbf{T}_{\mathcal{B}}$, where $\mathbf{T}_{\mathcal{A}}\neq\mathbf{T'}_{\mathcal{A}}$, we have the $k$-th undelayed bit-channel capacity, when $k\in \mathcal{B}$, as
		\begin{align}\label{eq:thm_1_2}
		&C^{[\mathbf{T}_{\mathcal{A}},\mathbf{T}_{\mathcal{B}}]}_{k,\text{DBICM}}(M\text{-QAM})=C^{[\mathbf{T'}_{\mathcal{A}},\mathbf{T}_{\mathcal{B}}]}_{k,\text{DBICM}}(M\text{-QAM}) \notag \\
		= & C^{\mathbf{T}_{\mathcal{B}}}_{(k-\frac{m}{2}),\text{DBICM}}(\sqrt{M}\text{-PAM}).
		\end{align}
	\end{thm}

	\begin{proof}
		See Appendix \ref{pf:proof_superposition}.
	\end{proof}

	Theorem \ref{thm:superposition} shows that $\mathbf{T}_{\mathcal{A}}$ and $\mathbf{T}_{\mathcal{B}}$ only affect the bit-channel capacities for coded bits in group $\mathcal{A}$ and $\mathcal{B}$ independently. In other words, the bit-channel capacity for any real part bit in $\mathcal{A}$ is only affected by the delay scheme of other real part bits and it is not affected by the imaginary bits, whether they are delayed or not delayed. The same is true for the bit-channel capacity of any imaginary part bit in $\mathcal{B}$. For example, consider two Gray labeled $64$-QAM DBICM systems with $\mathbf{T}=[1,0,0,0,0,1]$ and $\mathbf{T}=[1,0,0,1,0,1]$ and a Gray labeled $8$-PAM DBICM system with $\mathbf{T}=[1,0,0]$. Their second bit share identical bit-channel capacity as shown in Fig. \ref{fig:PAM_QAM_same}. 
	\begin{figure} [h]
		\centering
		\includegraphics[width=0.45\textwidth]{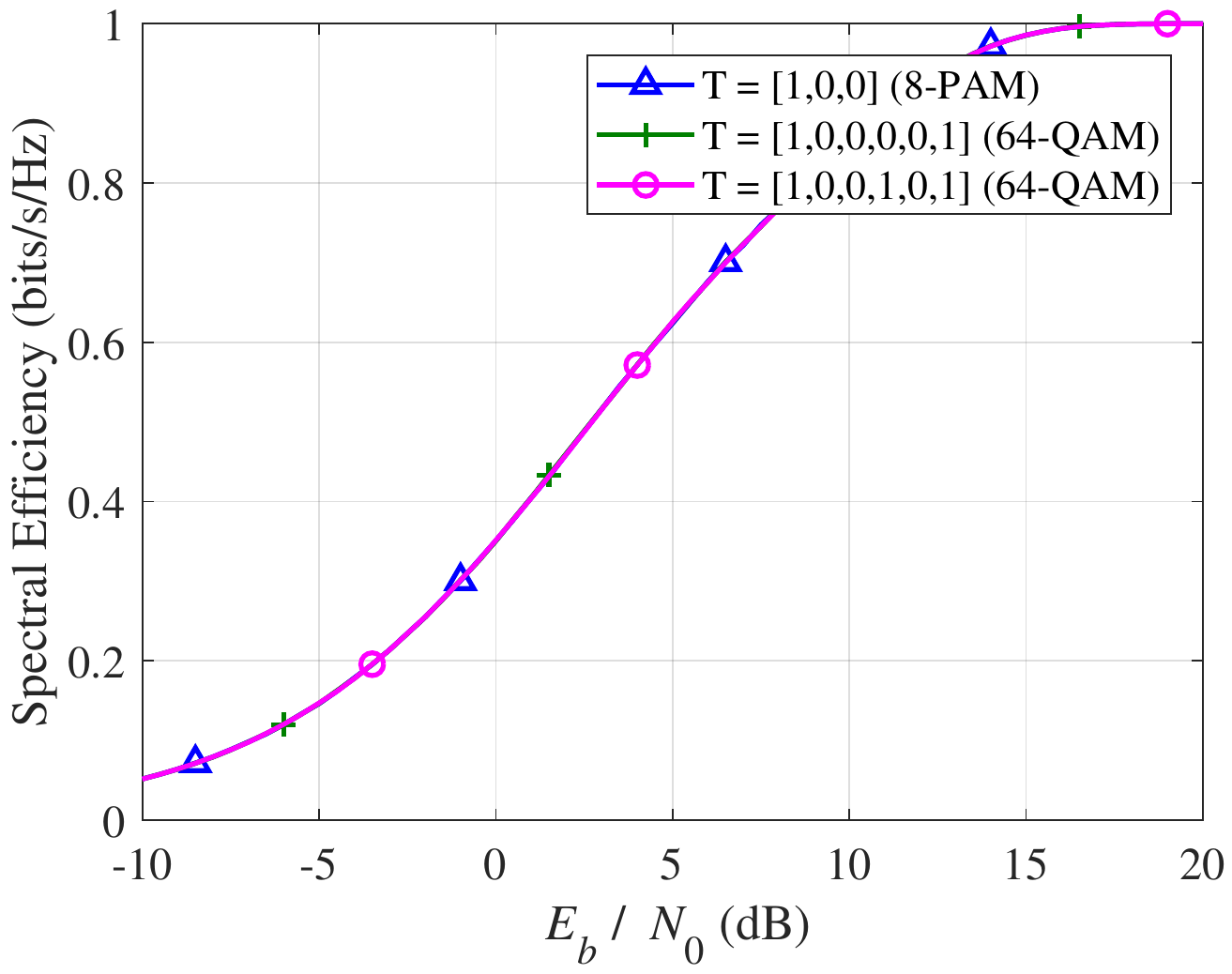}
		\caption{Bit capacity comparison for the $2$nd bit in $8$-PAM DBICM and $64$-QAM DBICM, while the first three bits share identical delay schemes in all three DBICM systems.} 
		\label{fig:PAM_QAM_same}
	\end{figure}	
	
	Furthermore, Theorem \ref{thm:superposition} also indicates that any coded bit in an $M$-QAM DBICM with delay scheme $\mathbf{T}=[\mathbf{T}_{\mathcal{A}},\mathbf{T}_{\mathcal{B}}]$ has identical bit-channel capacities as its corresponding coded bit in $\sqrt{M}$-PAM DBICM with delay scheme $\mathbf{T}_{\mathcal{A}}$ or $\mathbf{T}_{\mathcal{B}}$. Therefore, the bit-channel capacities in an $M$-QAM DBICM can be computed from two $\sqrt{M}$-PAM DBICM. In addition, when none of the coded bits in group $\mathcal{A}$ or $\mathcal{B}$ are delayed, we have the following corollary for the bit-channel capacities.
	\begin{cor}\label{cor:from_thm1}
		With none of the coded bits in group $\mathcal{A}$ being delayed, the $k$-th bit-channel capacity, when $k\in\mathcal{A}$, can be expressed as
		\begin{align}
		&C^{[\mathbf{0}_{1,\frac{m}{2}},\mathbf{T}_{\mathcal{B}}]}_{k,\text{DBICM}}(M\text{-QAM}) = C_{k,\text{BICM}}(M\text{-QAM}) \notag \\
		= & C_{k,\text{BICM}}(\sqrt{M}\text{-PAM}).
		\end{align}
		Alternatively, with none of the coded bits in group $\mathcal{B}$ being delayed, the $k$-th bit-channel capacity, when $k\in\mathcal{B}$, can be expressed as
		\begin{align}
		&C^{[\mathbf{T}_{\mathcal{A}}, \mathbf{0}_{1,\frac{m}{2}}]}_{k,\text{DBICM}}(M\text{-QAM}) =  C_{k,\text{BICM}}(M\text{-QAM}) \notag \\
		= & C_{(k-\frac{m}{2}),\text{BICM}}(\sqrt{M}\text{-PAM}).
		\end{align}
	\end{cor}Corollary \ref{cor:from_thm1} guarantees that to obtain bit-channel capacity improvements from coded bits in both group $\mathcal{A}$ and $\mathcal{B}$, one needs to delay at least one coded bit from each group. Since symmetric bits result in identical bit-channel capacities in uniform Gray labeled $M$-QAM BICM, it is worthwhile to investigate whether delaying symmetric bits in uniform Gray labeled $M$-QAM DBICM also result in the same capacity. In the following, we first define the symmetric delay schemes, then establish a theorem on the property of DBICM symmetric delay schemes.
	\begin{definition}\label{def:symmetric_schemes}
			Let $\mathbf{T}=[\mathbf{T}_{\mathcal{A}}, \mathbf{T}_{\mathcal{B}}]$ and $\mathbf{T'}=[\mathbf{T}_{\mathcal{B}}, \mathbf{T}_{\mathcal{A}}]$. We call this pair of delay schemes $\mathbf{T}$ and $\mathbf{T'}$, in uniform Gray labeled $M$-QAM DBICM systems, symmetric delay schemes.
	\end{definition}

	\begin{thm} \label{thm:D_D'_symmetric}
		The capacities of two uniform Gray labeled $M$-QAM DBICM systems under symmetric delay schemes which are defined in Definition \ref{def:symmetric_schemes} are the same. That is,
		\begin{equation}
		C^{\mathbf{T}}_{\text{DBICM}}(M\text{-QAM}) = C^{\mathbf{T}'}_{\text{DBICM}}(M\text{-QAM}).
		\end{equation}
	\end{thm}

	\begin{proof}
		See Appendix \ref{pf:proof_D_D'_symmetric}.
	\end{proof} 
	Theorems \ref{thm:superposition} and \ref{thm:D_D'_symmetric} show that the effects from delay coded bits in group $\mathcal{A}$ and $\mathcal{B}$ are independent and symmetric. Besides, in uniform Gray labeled BICM, two successive PAM symbols in real Gaussian noise are equivalent to a QAM symbol in circular symmetric complex Gaussian noise in BICM schemes \cite{1021039}. Therefore, following Theorems \ref{thm:superposition} and \ref{thm:D_D'_symmetric}, as well as the nature of uniform Gray labeled QAM, we conclude the DBICM capacity relationship between $M$-QAM and its underlying $\sqrt{M}$-PAMs as the following proposition.
	\begin{proposition}\label{prop:p}
			For a uniform Gray labeled $M$-QAM DBICM system under a delay scheme $\mathbf{T} = [\mathbf{T}_{\mathcal{A}},\mathbf{T}_{\mathcal{B}}]$, where $\mathbf{T}_{\mathcal{A}}$ and $\mathbf{T}_{\mathcal{B}}$  are defined in Theorem \ref{thm:superposition}, we have
			\begin{equation}
			C^{\mathbf{T}}_{\text{DBICM}}(\text{$M$-QAM}) = C^{\mathbf{T}_{\mathcal{A}}}_{\text{DBICM}}(\text{$\sqrt{M}$-PAM})+C^{\mathbf{T}_{\mathcal{B}}}_{\text{DBICM}}(\sqrt{M}\text{-PAM}).  
			\end{equation}
	\end{proposition}
	\begin{proof}
		See Appendix \ref{pf:proof_prop_1}.
	\end{proof} Proposition \ref{prop:p} shows that the capacity of a uniform Gray labeled $M$-QAM is the sum of the capacities of its underlying uniform Gray labeled $\sqrt{M}$-PAM DBICM systems. For example, a uniform Gray labeled $64$-QAM DBICM system with a delay scheme $\mathbf{T}= [0, 1, 0, 1, 1, 0]$ can be written as follows,
	\begin{equation}
	C^{[0, 1, 0, 1, 1, 0]}_{\text{DBICM}}(\text{$64$-QAM}) = C^{[0, 1, 0]}_{\text{DBICM}}(\text{$8$-PAM})+C^{[1, 1, 0]}_{\text{DBICM}}(8\text{-PAM}).  
	\end{equation}
	Therefore, the capacity analysis for DBICM $M$-QAM is reduced to that of the underlying DBICM $\sqrt{M}$-PAM under the delay scheme $\mathbf{T}_{\mathcal{A}}$ and $\mathbf{T}_{\mathcal{B}}$. Then, the optimum delay scheme $\mathbf{T}^*$ for a DBICM $M$-QAM with a given $E_s/N_0$ is that
	\begin{align}\label{eq:18}
	&C_\text{DBICM}^{\mathbf{T}^*}(M\text{-QAM}) \notag \\
	 = & \max_{\mathbf{T}_{\mathcal{A}} \in \mathbb{F}_2^{ \frac{m}{2}}}\{{C_{\text{ DBICM}}^{\mathbf{T}_{\mathcal{A}}}}(\sqrt{M}\text{-PAM})\} + \max_{\mathbf{T}_{\mathcal{B}} \in \mathbb{F}_2^{\frac{m}{2}}}\{{C_{\text{ DBICM}}^{\mathbf{T}_{\mathcal{B}}}}(\sqrt{M}\text{-PAM})\}.
	\end{align}
	
	In what follows, under the constrain of $T_{\text{max}}=1$, we provide a method to find the optimal delay scheme $\mathbf{T}^*$ with 2 steps.
	\begin{enumerate}
		\item Find the optimal delay scheme $\mathbf{T}_{\mathcal{A}}^{*}$ among all possible delay schemes for the underlying DBICM $\sqrt{M}$-PAM that gives the largest capacity, such that
		\begin{equation}
		C_\text{DBICM}^{\mathbf{T}^*_{\mathcal{A}}}(\sqrt{M}\text{-PAM}) = \max_{\mathbf{T}_{\mathcal{A}} \in \mathbb{F}_2^{\frac{m}{2}}}\{{C_{\text{ DBICM}}^{\mathbf{T}_{\mathcal{A}}}}(\sqrt{M}\text{-PAM})\}.
		\end{equation}
		\item The optimal delay scheme $\mathbf{T}^{*}$ for DBICM $M$-QAM is then obtained as 
		\begin{equation}
		\mathbf{T}^* = [\mathbf{T}_{\mathcal{A}}^*, \mathbf{T}_{\mathcal{B}}^*] = [\mathbf{T}_{\mathcal{A}}^*, \mathbf{T}_{\mathcal{A}}^*],
		\end{equation} where choosing $\mathbf{T}_{\mathcal{B}}^* = \mathbf{T}_{\mathcal{A}}^*$ leads to the largest capacity of DBICM $M$-QAM according to Eq. (\ref{eq:18}).
	\end{enumerate}
	 
	We point out that the computational complexity for finding the optimal delay scheme for a DBICM $M$-QAM is reduced by using our method. We first consider the situation where the largest number of delayed time slot is $T_{\text{max}}=1$. For a given $E_s/N_0$, the computational complexity for exhaustive search is with $\mathcal{O}(M^2)$ because the constellation size is $M$ and there are $M$ possible delay schemes in total. In contrast, our method has a reduced complexity of $\mathcal{O}(M)$ as for a $\sqrt{M}$-PAM there are $\sqrt{M}$ delay schemes in total. Hence, even finding the optimal delay scheme for $1024$-QAM DBICM is still practical by using the proposed method. In general, the number of possible delay schemes can be magnified as the largest number of delayed time slots $T_{\text{max}}$ increases. For example, there are $(T_{\text{max}}+1)^{\log_2M}$ and $(\sqrt{T_{\text{max}}+1})^{\log_2M}$ possible delay schemes in total for $M$-QAM and $\sqrt{M}$-PAM DBICM, respectively. Therefore, this method has greatly reduced the computational complexity from $\mathcal{O}(M(T_{\text{max}}+1)^{\log_2M})$ to $\mathcal{O}(\sqrt{M}(\sqrt{T_{\text{max}}+1})^{\log_2M})$.
	\subsection{Design Examples}
	\begin{figure} [h]
		\centering
		\includegraphics[width=0.45\textwidth]{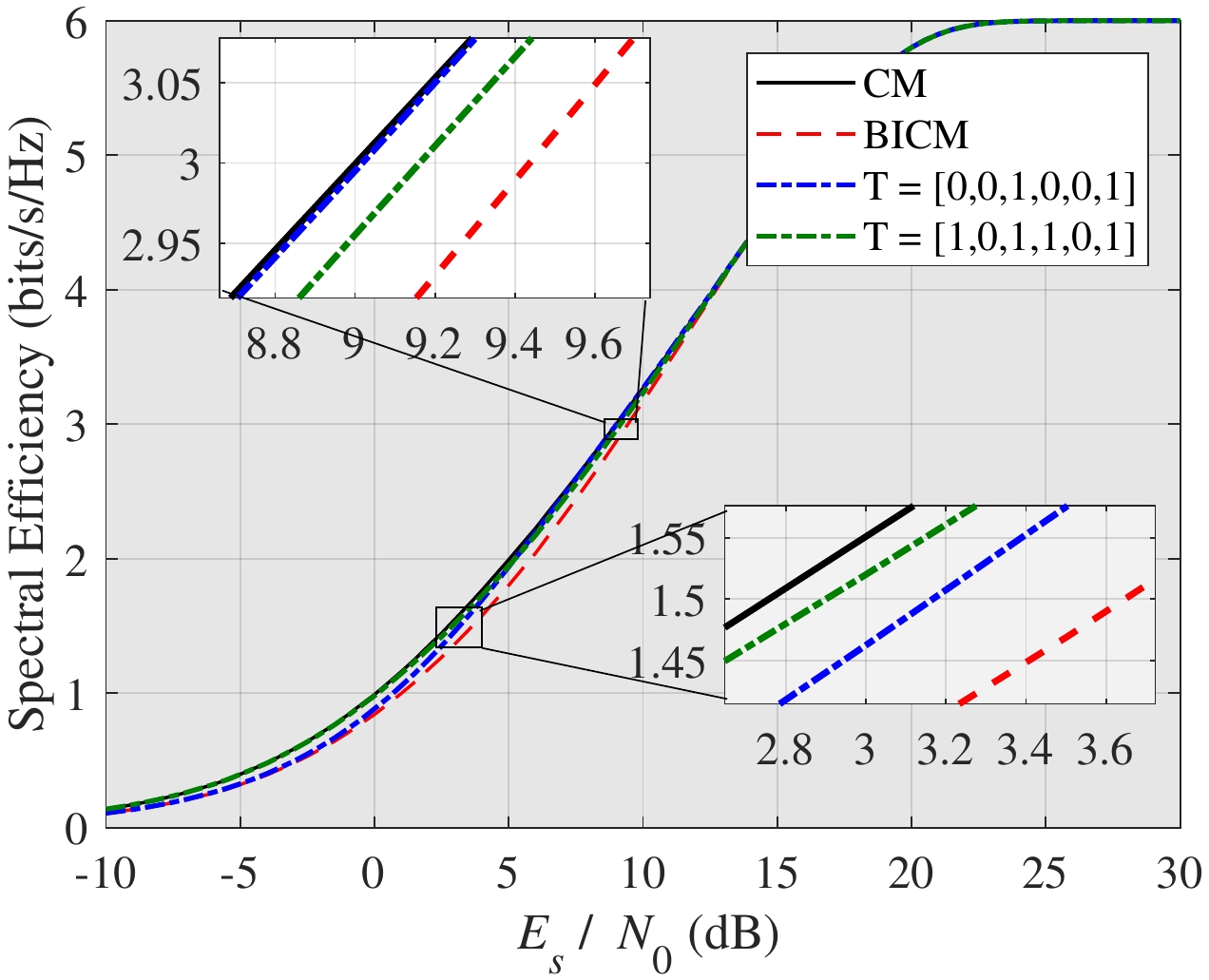}
		\caption{Capacity comparison between DBICM, CM and BICM.} \label{fig:cap_64_dbicm_opt}
\end{figure}
	In this section, we apply the proposed method to find the optimum delay scheme for a variety of modulation levels and code rates. We consider a uniform Gray labeled $64$-QAM DBICM system. First, we compute the capacity of $8$-PAM DBICM with all distinct delay schemes. Note that the delay schemes $\mathbf{T}_{\mathcal{A}} = [0,0,0]$ and $\mathbf{T}_{\mathcal{A}} = [1,1,1]$ are excluded because the bit-channel capacities under these delay schemes are the same for BICM. When the spectral efficiencies are $0.75$ and $1.5$ bits/symbol, the optimal delay scheme for $8$-PAM DBICM is $\mathbf{T}_{\mathcal{A}}^*=[1,0,1]$ and $\mathbf{T}_{\mathcal{A}}^*=[0,0,1]$, respectively. Therefore, we obtain the optimal delay schemes $\mathbf{T}^*=[1,0,1,1,0,1]$ and $\mathbf{T}^*=[0,0,1,0,0,1]$ for $64$-QAM DBICM at spectral efficiencies of $1.5$ and $3$ bits/s/Hz, corresponding to code rates $1/4$ and $1/2$, respectively. 
	
	In Fig. \ref{fig:cap_64_dbicm_opt}, we plot the 64-QAM DBICM capacities under the optimal delay schemes and compare them with the constellation constrained capacity and BICM capacity of 64-QAM. At rate $1/4$, DBICM with the optimized delay scheme $\mathbf{T}=[1,0,1,1,0,1]$ is within $0.1$ dB of the constellation constrained capacity at a spectral efficiency of $1.5$ bits/s/Hz and is $0.7$ dB better than its BICM counterpart. The DBICM capacity of another delay scheme $\mathbf{T}=[0,0,1,0,0,1]$ optimized for rate $1/2$ is even within 0.01 dB of the constellation constrained capacity at a spectral efficiency of $3$ bits/s/Hz, while obtaining a capacity improvement of around $0.45$ dB over its BICM counterpart. We would point out that when the code rate is higher than $1/2$, the DBICM and BICM capacities are very close to that of the constellation constrained capacity.

	\begin{table*}[h]
		\centering
		\caption{Optimum delay scheme $\mathbf{T}^*$ for uniform Gray labeled $M$-QAM DBICM with code Rate $1/4$, $1/3$, $2/5$ and $1/2$.~~~~~~} \label{tab:opt_delay_scheme}
		\begin{tabular}{|c|c|c|c|c|}
			\hline
			Modulation                & Rate & Optimal Delay Scheme & Gap to CM (dB) & Gain over BICM (dB) \\ \hline
			\multirow{4}{*}{16-QAM}   & 1/4  &  [0, 1, 0, 1]                    &       0         &             0.55        \\ \cline{2-5} 
			& 1/3  &  [0, 1, 0, 1]                    &       0         &             0.4        \\ \cline{2-5} 
			& 2/5  &   [0, 1, 0, 1]                   &     0           &              0.3       \\ \cline{2-5} 
			& 1/2  &  [0, 1, 0, 1]                    &      0          &      0.2               \\ \hline
			\multirow{4}{*}{64-QAM}   & 1/4  &  [1, 0, 1, 1, 0, 1]                    &    0.15           &        0.7             \\ \cline{2-5} 
			& 1/3  &  [0, 1, 0, 0, 1, 0]                    &       0.15         &             0.6        \\ \cline{2-5} 
			& 2/5  &    [0, 0, 1, 0, 0, 1]                  &    0.1            &    0.55                 \\ \cline{2-5} 
			& 1/2  &    [0, 0, 1, 0, 0, 1]                   &       0.01         &       0.45              \\ \hline
			\multirow{4}{*}{256-QAM}  & 1/4  &  [0, 0, 1, 1, 0, 0, 1, 1]                    &     0.3           &        0.65             \\ \cline{2-5} 
			& 1/3  &  [1, 1, 0, 1, 1, 1, 0, 1]                    &     0.25           &      0.65               \\ \cline{2-5} 
			& 2/5  &  [0, 0, 1, 1, 0, 0, 1, 1]                    &     0.25           &      0.65               \\ \cline{2-5} 
			& 1/2  &  [0, 0, 0, 1, 0, 0, 0, 1]                    &     0.15           &                 0.6    \\ \hline
			\multirow{4}{*}{1024-QAM} & 1/4  & [1, 1, 0, 1, 1, 1, 1, 0, 1, 1]                     &     0.25           &      0.85               \\ \cline{2-5} 
			& 1/3  &  [1, 0, 0, 1, 1, 1, 0, 0, 1, 1]                    &      0.25          &         0.8            \\ \cline{2-5} 
			& 2/5  &  [0, 0, 0, 1, 1, 0, 0, 0, 1, 1]                    &      0.25          &         0.8            \\ \cline{2-5} 
			& 1/2  &  [0, 0, 0, 1, 1, 0, 0, 0, 1, 1]                    &    0.25            &    0.65                 \\ \hline
		\end{tabular}
	\end{table*}
	In Table \ref{tab:opt_delay_scheme}, we list the optimal delay schemes for DBICM systems with Gray labeled square $16$-QAM, $64$-QAM, $256$-QAM and $1024$-QAM constellations at four code rates, $\frac{1}{4}$, $\frac{1}{3}$, $\frac{2}{5}$, and $\frac{1}{2}$. In addition, we also listed the capacity gap between the DBICM with optimal delay scheme and the constellation constrained capacity, as well as the corresponding capacity gain over BICM in Table \ref{tab:opt_delay_scheme}. It can be observed that the capacity of DBICM becomes very close to that of the constellation constrained capacity when the code rate increases and the gain of DBICM over BICM increases with constellation size $M$. One may also notice that the capacity improvement from BICM to DBICM is not large in the high code rate region for uniform Gray labeled QAM. However, this may not be the case for other constellations as previous work in \cite{DBICM_cap} shows that DBICM with quasi-Gray labeled half $16$-QAM can achieve up to 0.4 dB capacity gain over BICM at rate 3/4.
	\section{Irregular LDPC finite code design for DBICM}
	In this section, we provide a code construction and optimization method for DBICM with irregular LDPC codes. Recall that UEP exists in high order modulation systems, where not all bit-channels are with equal capacities. When designing irregular LDPC codes for DBICM, it is important to take UEP into account. In particular, we take special consideration of the connection between variable nodes (VNs) with different degrees and various bit-channels, which we refer to as the channel assignment. The conventional LDPC code design needs to optimize the degree distributions for the VNs and check nodes (CNs). However, the code design in this paper needs to consider both the degree distributions and the channel assignment. As in conventional LDPC codes \cite{1365488}, our code design adopted a concentrated CN degree $d_{\text{c}}$, such that for $j\in\{1,\ldots,N(1-R)\}$, the $j$-th CN, denoted by $c_j$, is connected with $d_{\text{c}}$ VNs. Then, we apply differential evolution (DE) \cite{DE} to optimize both the VN degree distributions and channel assignments, subject to the lowest decoding threshold. The decoding threshold is computed via PEXIT chart, appreciating its high accuracy in estimating the decoding threshold for high order modulated LDPC codes \cite{7339431}.

	\subsection{Protograph-EXIT Chart Based Code optimization Algorithm}\label{subsec:code_opt}

	To model the channel assignment, we first classify bit-channels into different bit-channel types according to their corresponding bit-channel capacities, such that all the bit-channel capacities are identical within each bit-channel type. For example, in a uniform Gray labeled $64$-QAM DBICM with a  delay scheme $\mathbf{T}=[0,0,1,0,0,1]$, we classify six bit-channels into three bit-channel types regarding their capacities. We show the mapping between bit-channels and bit-channel types in Table \ref{tab:type}. In this example, type $i$ bit-channels have higher capacities than type $i+1$ bit-channels, for $i \in \{0,1\}$ and for all SNR.

	\begin{table}[h]
		\centering
		\caption{An example of bit-channel classification in DBICM scheme.~~~~~~} \label{tab:type}
		\begin{tabular}{|c|c|c|c|}
			\hline
			Bit-channel Type             & $0$ & $1$ & $2$ \\ \hline
			{Bit-channel} 				 & $\{0,3\}$ & $\{1,4\}$ & $\{2,5\}$ \\ \hline
		\end{tabular}
	\end{table} 
	
	Let the VN degree distribution vector be $\bm{\lambda}=[\lambda_{j}]_{j=1}^{V}$, where the $\lambda_{j}$ represents the fraction of all edges connected to degree-$j$ VNs and $V$ is the maximum number of VN degrees. Thus, $\lambda_{j}$ should satisfy the following constraints:
	\begin{equation}\label{eq:lambdaj}
	0\le\lambda_{j}\le1, \quad \text{and} \quad \sum_{j=1}^{V}\lambda_j=1.
	\end{equation} 
	
	Let us define a channel assignment matrix $\mathbf{P}\in\mathbb{R}^{S\times V}$, where $S$ is the number of bit-channel types, and $p_{i,j}\in\mathbf{P}$ is the element in the $i$-th row and $j$-th column representing the fraction of degree-$j$ VNs assigned to the $i$-th bit-channel type. Note that the channel assignment matrix $\mathbf{P}$ is a generalized version of that in \cite{junyi_regandirr} where matrix $\mathbf{P}$ only has two rows since all bit-channels are classified into two types, i.e., the reliable and the unreliable channels. This generalization can improve the accuracy of decoding threshold analysis since the capacity differences between all bit-channels are considered via the channel assignment matrix $\mathbf{P}$ in this paper. This is important for high level modulations.
	
	Let $m_i$ be the number of bit-channels of channel type $i\in\{0,..., S-1\}$. The channel assignment matrix $\mathbf{P}$ can be generated from $\bm{\lambda}$ with the following constraints:
	\begin{subnumcases}{}
		0 \leq p_{i,j} \leq 1,  \label{eq:pvn_a}\\
		\sum\nolimits_{i=0}^{S-1}\sum\nolimits_{j=1}^{V}p_{i,j} = 1, \label{eq:pvn_b}\\
		\sum\nolimits_{i=0}^{S-1}p_{i,j}=\lambda_j, \forall j\in\{1,\ldots,V\}, \label{eq:pvn_c}\\
		\sum\nolimits_{j=1}^{V}p_{i,j}=\frac{m_i}{m}, \forall i\in\{1,\ldots,S\}. \label{eq:pvn_d}
	\end{subnumcases}

	Eqs. (\ref{eq:pvn_a}-\ref{eq:pvn_b}) specify the range for each element $p_{i,j}$ and its general constraint for $p_{i,j}\in\mathbf{P}$. Eq. (\ref{eq:pvn_c}) specifies the relationship between different channel types of the same VN degree $j$ in $\mathbf{P}$ and $\lambda_j$. The relationship between the edges of a specific channel type and the number of bit-channels in the specific type is described by Eq. (\ref{eq:pvn_d}). Although all the bit-channels need to be used in the coding scheme, randomly choosing a channel assignment could lead to performance degradation.
	
	Now, with these constraints for $\bm{\lambda}$ and $\mathbf{P}$, we can construct LDPC codes for DBICM and compute its decoding threshold $\theta$ via PEXIT chart. In the following, we optimize $\bm{\lambda}$ and $\mathbf{P}$ with a minimum $\theta$. Note that, since there are $V$ unknown variables, one equation and one inequality in Eq. (\ref{eq:lambdaj}), there are $V-1$ free variables, which are needed to be optimized for $\bm{\lambda}$. Similarly, for $p_{i,j}\in\mathbf{P}$ in Eqs. (\ref{eq:pvn_a}-\ref{eq:pvn_d}), $S\times V-S-V-1$ free variables are needed to be optimized. In general, these free variables can be any real numbers within their associated ranges, leading to a huge optimization space. In this paper, we resort to DE to optimize $\bm{\lambda}$ and $\mathbf{P}$ for the proposed DBICM schemes. Ideally, the VN distribution $\bm{\lambda}$ and channel assignment matrix $\mathbf{P}$ should be optimized jointly via nesting the optimization for $\bm{\lambda}$ and $\mathbf{P}$ together, which means that DE is performed recursively. The recursive use of DE results in a very high optimization complexity even for a moderate population size in each DE. Hence, as a practical alternative, a two-step cascaded DE is used to optimize $\bm{\lambda}$ and $\mathbf{P}$ in each step individually. In both DE steps, the inputs are composed of $N_1$ optimization object candidates. These $N_1$ candidates are evolved through $N_2$ generations to reduce the decoding threshold.
		\begin{table}[h]\label{table:algorithm 1} 
		\begin{algorithm} [H]              
			\renewcommand\thealgorithm{1}
			\centering
			\caption{Code optimization Algorithm}
			\label{algorithm:proposed_opt_alg}    
			\begin{algorithmic} [1] 
				\REQUIRE $C^{\mathbf{T^*}}_{{i},\text{DBICM}}, i\in\{0,1,...,m-1\}$.
				\ENSURE $\mathbf{H}^{*}$
				\PRINT Set $\theta^{*} = \infty$, randomly generate $\mathcal{C}_{\bm{\lambda}}^{0}$, where $\bm{\lambda}$ follows the constraints in Eq. (\ref{eq:lambdaj}).
				\FOR {$k=0:N_2 -1$}
				\FOR {$l = 1:N_1$} 
				\STATE Generate $\mathbf{H}_{l}^{k}$ from ${\bm{\lambda}}_{l}^{k}$ via the conventional PEG algorithm\cite{PEG}.
				\STATE Calculate $\theta_{l}^{k}$ for $\mathbf{H}_{l}^{k}$ using PEXIT chart \cite{PEXIT_chart} with given $C^{\mathbf{T}^*}_{i,\text{DBICM}}$.
				\IF {$k>0$}
				\IF {$\theta_{l}^{k-1}<\theta_{l}^{k}$}
				\STATE $\bm{\lambda}_{l}^{k} = \bm{\lambda}_{l}^{k-1}$.
				\ENDIF
				\IF {$\theta_{l}^{k} < \theta^{*}$}
				\STATE $\theta^{*} = \theta_{l}^{k}$, $\bm{\lambda}^{*} = \bm{\lambda}_{l}^{k}$, $\mathbf{H}^{*} = \mathbf{H}_{l}^{k}$.
				\ENDIF
				\ENDIF
				\ENDFOR
				\STATE Generate $\mathcal{C}_{\bm{\lambda}}^{k+1}$ from $\mathcal{C}_{\bm{\lambda}}^{k}$ via mutation and recombination in DE \cite{DE}.
				\ENDFOR 
				\PRINT Randomly generate $\mathcal{C}_{\mathbf{P}}^{0}$ from $\bm{\lambda}^{*}$. $\forall \mathbf{P}\in \mathcal{C}_{\mathbf{P}}^{0}$ satisfies Eqs. (\ref{eq:pvn_a}-\ref{eq:pvn_d}).
				\STATE Change the optimization object $\bm{\lambda}$ to $\mathbf{P}$ and replace the code construction algorithm in Step 4 by Algorithm \ref{algorithm:proposed_code_alg}, and then repeat Steps 2-16.
			\end{algorithmic}
		\end{algorithm}
	\end{table} 
	
	In the following, we introduce the parameters used in the code optimization algorithm. We denote the index of the current generation by $k$ in the DE, and the candidate sets for $\bm{\lambda}$ and $\mathbf{P}$ in the $k$-th generation by $\mathcal{C}^k_{\bm{\lambda}}$ and $\mathcal{C}^k_{\mathbf{P}}$, respectively, where $k\in\{0,\ldots,N_2\}$, $k=0$ represents the randomly initialized generation and $N_2$ denotes the maximum number of generations allowed for DE. For each generation, the population size is $N_1$. Let ${\bm{\lambda}}_{l}^{k}$ and $\mathbf{P}_{l}^{k}$, $l \in \{1,...,N_1\}$, be the $l$-th candidate in the $k$-th generation of $\mathcal{C}_{\bm{\lambda}}^{k}$ and $\mathcal{C}_{\mathbf{P}}^{k}$, respectively. Let $\mathbf{H}_{l}^{k}$ represent the parity-check matrix associated with ${\bm{\lambda}}_{l}^{k}$ or $\mathbf{P}_{l}^{k}$ and the corresponding decoding threshold $\theta_{l}^k$ is computed from PEXIT chart. Let $\theta^*$, $\bm{\lambda}^*$, $\mathbf{P}^*$ and $\mathbf{H}^*$ be the optimized decoding threshold, VN degree distribution, channel assignment matrix and parity-check matrix, respectively. The proposed code optimization algorithm is presented in \textbf{Algorithm 1}. 

	In \textbf{Algorithm 1}, Steps 1 to 16 optimize $\bm{\lambda}$ and Steps 17 to 18 optimize $\mathbf{P}$. These two parts are cascaded to obtained $\mathbf{H}^*$ which has the lowest decoding SNR threshold $\theta^*$. Here $\mathbf{H}^*$ should satisfy $\mathbf{P}^*$ and it can be constructed via the proposed PEG-like code construction algorithm (\textbf{Algorithm \ref{algorithm:proposed_code_alg}}), which will be presented next.

	\subsection{Constrained PEG-like Code Construction Algorithm}
	To construct LDPC codes with a minimized decoding threshold for high order modulated systems, a well designed channel assignment matrix $\mathbf{P}$ is required to be satisfied in the code construction. While the conventional PEG algorithm can only satisfy the degree distribution $\bm{\lambda}$ but not the $\mathbf{P}$, we propose a constrained PEG-like algorithm for constructing LDPC codes that satisfy both the degree distribution $\bm{\lambda}$ and the channel assignment matrix $\mathbf{P}$.
	
	Let us consider an LDPC code ensemble with VN degrees vary from $d_{v_\text{1}}$ to $d_{v_\text{V}}$ and CN degree is fixed at $d_c$. We define VN degree sequences $\mathbf{g}_i=[g_{i,j}]_{j=1}^{Nm_i/m}$, where $g_{i,j}$ is the $j$-th element of $\mathbf{g}_i$. The VN degree sequences $\mathbf{g}_i$ is generated from the channel assignments $\mathbf{P}$ as 
	\begin{align} \label{eq:g_i}
	&\mathbf{g}_i = [\overbrace{\strut \underbrace{\strut d_{v_\text{1}},\cdots,d_{v_\text{1}}}_{n p_{i,1}},\underbrace{\strut d_{v_\text{2}},\cdots,d_{v_\text{2}}}_{n p_{i,2}},\cdots, \underbrace{\strut d_{v_\text{V}},\cdots,d_{v_\text{V}}}_{n p_{i,V}}}^{Nm_i/m}].
	\end{align} For $i\in\{0,...,S-1\}$, $\mathbf{g}_i$ lists the degrees of all VNs that are assigned with the $i$-th bit-channel type. Note that the length of $\mathbf{g}_i$ may differ for different bit-channel type $i$. We introduce a temporary counter vector, denoted by $\mathbf{a}=[a_i]_{i=0}^{S-1}, a_i\in\{1,\ldots,Nm_i/m\}$, to count the times of each individual $\mathbf{g}_i$ being used by the constrained PEG-like algorithm, such that the $j$-th element $g_{i,j}\in\mathbf{g}_{i}$ is assigned to the $a_i$-th VN in the $i$-th bit-channel type.
	
	The idea of the proposed constrained PEG-like code construction algorithm is as follows. For a given code length $N$, we map VNs to bit-channels by using the continuous bit mapping. Specifically, for $i\in\{1,\ldots,N\}$, the $i$-th VN, denoted by $v_i$, is mapped to the $(i-1)_m$-th bit-channel, where $(i-1)_m \triangleq (i-1)\mod{m}$. Furthermore, by following the bit-channel type classification in Section \ref{subsec:code_opt}, we define a mapping function that maps the $j$-th bit-channel to the $k$-th bit-channel type i.e., $\phi(j)=k$, where $j\in\{0,\ldots,m-1\}$ and $k\in\{0,\ldots,S-1\}$. Then, VN $v_i$ is implicitly mapped to the $\phi((i-1)_m)$-th bit-channel type. For example, in a uniform Gray labeled $64$-QAM DBICM with a delay scheme $\mathbf{T}=[0,0,1,0,0,1]$, $v_7$ is mapped to the $0$-th bit-channel and the $\phi(0)$-th bit-channel type, where $\phi(0) = 0$ according to Table \ref{tab:type}. We assign $v_i$ with degree-$a_{\phi((i-1)_m)}$ from $\mathbf{g}_{\phi((i-1)_m)}$. At last, a constrained PEG-like algorithm connects all the VNs and CNs following the degree distribution and the channel assignment $\mathbf{P}$. In this way, $\mathbf{P}$ is imposed as a constraint in constructing LDPC codes for DBICM and BICM. The details of the constrained PEG-like algorithm are presented in \textbf{Algorithm 2}: 

	\begin{table}[h]\label{table:algorithm2} 
		\begin{algorithm} [H]              
			\renewcommand\thealgorithm{2}
			\caption{Constrained PEG-like Algorithm}
			\label{algorithm:proposed_code_alg}    
			\begin{algorithmic} [1] 
				\REQUIRE $\mathbf{P}$, $N$, $d_{\text{c}}$, $R$, $m$, and $[m_i]_i\in\{0,\ldots,S-1\}$.
				\ENSURE $\mathbf{H}$. 
				\PRINT $\mathbf{H} = \mathbf{0}_{N(1-R),N}$ and $\mathbf{a}=\mathbf{0}$.
				\FOR {$i=0:S-1$}
				\STATE Generate $\mathbf{g}_i$ from $\mathbf{P}$ via Eq. (\ref{eq:g_i}). 
				\ENDFOR
				\FOR {$i=1:N$}
				\STATE $a_{\phi((i-1)_m)} = a_{\phi((i-1)_m)} + 1.$
				\FOR {$l=1:g_{\phi((i-1)_m),a_{\phi((i-1)_m)}}$}
				\STATE Find the most distant CNs, connected with less than $d_{\text{c}}$ VNs, from $v_i$. 
				\IF {Multiple CN candidates exist} 
				\STATE Randomly choose a CN $c_j$, among the available CN candidates.
				\STATE $H_{i,j} = 1$.
				\ENDIF
				\ENDFOR
				\ENDFOR 
			\end{algorithmic}
		\end{algorithm}
	\end{table}

	In \textbf{Algorithm 2}, Steps 2-4 generate $\mathbf{g}_i, i\in\{0,\ldots,S-1\},$ that satisfies the channel assignments $\mathbf{P}$. Then all $\mathbf{g}_i$ are applied to the PEG algorithm as the constraints to construct LDPC codes in Steps 5 to 14.

	\section{Numerical Results}
	
	In this section, we construct and optimize LDPC codes for uniform Gray labeled $M$-QAM DBICM, $M = \{16, 64\}$, with the optimal delay schemes listed in Table \ref{tab:opt_delay_scheme}. We also evaluate the BER performances of the designed codes via simulation. The LDPC code rates for the DBICM schemes are $1/4$, $2/5$ and $1/2$. The codeword lengths for $16$-QAM and $64$-QAM are $100,000$ and $120,000$, respectively. In particular, we set the maximum number of VN degree $V=10$, as a higher value of $V$ imposes larger optimization complexity but marginal performance improvement. The concentrated CN degrees of 4, 5 and 7 are applied for code rates $1/4$, $2/5$, and $1/2$, respectively. For the parameters employed in DE, the differential weight and crossover probability are set to $0.5$ by following \cite{DE}. We also set population number $N_{1}$ scaling with the size of $\bm{\lambda}$ or $\mathbf{P}$, such as $10(V-1)$ and $10(SV-1)$ for $\mathcal{C}_{\bm{\lambda}}$ and $\mathcal{C}_{\mathbf{P}}$, respectively. We set the number of generations in both DE steps for $\bm{\lambda}$ and $\mathbf{P}$ to $N_2=10$, as it is observed that the improvements of the designed codes' decoding thresholds are minor. We also design LDPC codes for BICM by using our code optimization algorithm and the constrained PEG-like algorithm for comparison benchmark.

	\begin{table*}[h]
			\caption{optimized $\mathbf{P}$ and the thresholds for uniform Gray labeled $16$-QAM DBICM and BICM with Code Rates $1/4$, $2/5$ and $1/2$.~~~~~~} \label{tab:pvn_16qam}
								\begin{tabular}{|c|c|c|c|c|c|c|c|c|c|c|c|c|}
					\hline
					R   & MOD   & Bit-channels & $d_{v_2}$      & $d_{v_3}$      & $d_{v_4}$      & $d_{v_5}$      & $d_{v_6}$      & $d_{v_7}$    & $d_{v_8}$      & $d_{v_9}$      & $d_{v_{10}}$      & Thresholds   ($E_b/N_0$)   \\ \hline
					\multirow{4}{*}{1/4} & \multirow{2}{*}{DBICM} & $\{0,2\}$    & 0.3866 & 0.0575 & 0.0000 & 0.0006 & 0.0003 & 0.0000 & 0.0000 & 0.0113 & 0.0436 & \multirow{2}{*}{$0.8398$} \\ \cline{3-12} 
					&                        & $\{1,3\}$    & 0.4020 & 0.0381 & 0.0000 & 0.0009 & 0.0001 & 0.0000 & 0.0002 & 0.0007 & 0.0580 & \\ \cline{2-13} 
					& \multirow{2}{*}{BICM}  & $\{0,2\}$    & 0.3580 & 0.0793 & 0.0000 & 0.0009 & 0.0042 & 0.0048 & 0.0005 & 0.0061 & 0.0462 & \multirow{2}{*}{$1.3672$}\\ \cline{3-12} 
					&                        & $\{1,3\}$    & 0.4149 & 0.0291 & 0.0003 & 0.0002 & 0.0000 & 0.0027 & 0.0005 & 0.0040 & 0.0484 & \\ \hline
					\multirow{4}{*}{2/5} & \multirow{2}{*}{DBICM} & $\{0,2\}$    & 0.3039 & 0.1868 & 0.0000 & 0.0000 & 0.0030 & 0.0038 & 0.0013 & 0.0009 & 0.0002 & \multirow{2}{*}{$1.8066$} \\ \cline{3-12} 
					&                        & $\{1,3\}$    & 0.3375 & 0.0601 & 0.0060 & 0.0002 & 0.0124 & 0.0057 & 0.0008 & 0.0081 & 0.0693 & \\ \cline{2-13} 
					& \multirow{2}{*}{BICM}  & $\{0,2\}$    & 0.4918 & 0.0000 & 0.0000 & 0.0017 & 0.0019 & 0.0000 & 0.0001 & 0.0037 & 0.0008 & \multirow{2}{*}{$2.0938$}\\ \cline{3-12} 
					&                        & $\{1,3\}$   & 0.1672 & 0.2141 & 0.0034 & 0.0005 & 0.0066 & 0.0000 & 0.0513 & 0.0567 & 0.0001 & \\ \hline
					\multirow{4}{*}{1/2} & \multirow{2}{*}{DBICM} & $\{0,2\}$    & 0.3579 & 0.0887 & 0.0000 & 0.0015 & 0.0000 & 0.0000 & 0.0004 & 0.0003 & 0.0512 & \multirow{2}{*}{$2.5703$}\\ \cline{3-12} 
					&                        & $\{1,3\}$    & 0.2623 & 0.1219 & 0.0000 & 0.0017 & 0.0016 & 0.0000 & 0.0193 & 0.0018 & 0.0913 & \\ \cline{2-13} 
					& \multirow{2}{*}{BICM}  & $\{0,2\}$    & 0.2457 & 0.1819 & 0.0029 & 0.0000 & 0.0001 & 0.0013 & 0.0004 & 0.0093 & 0.0583 & \multirow{2}{*}{$2.7266$}\\ \cline{3-12} 
					&                        & $\{1,3\}$    & 0.3464 & 0.0605 & 0.0035 & 0.0000 & 0.0002 & 0.0024 & 0.0001 & 0.0052 & 0.0817 & \\ \cline{1-13} 
				\end{tabular}
		\end{table*}

	The bit-channel assignments and their corresponding decoding threshold, for $16$-QAM and $64$-QAM, are listed in Table \ref{tab:pvn_16qam} and Table \ref{tab:pvn_64qam}, separately. It can be observed that the designed LDPC codes for DBICM have better decoding thresholds than the corresponding BICM. For instance, the designed LDPC codes for DBICM with the optimal delay scheme at rate $1/4$ can reduce decoding thresholds by $0.4674$ dB and $0.7129$ dB over BICM for $16$-QAM and $64$-QAM, respectively. Also, the decoding thresholds for all designed codes are within $0.8$ dB to their corresponding capacity limits.
	
	\begin{figure}[h]
		\centering
		\subfigure[16-QAM.]
		{
			\includegraphics[width=0.45\textwidth]{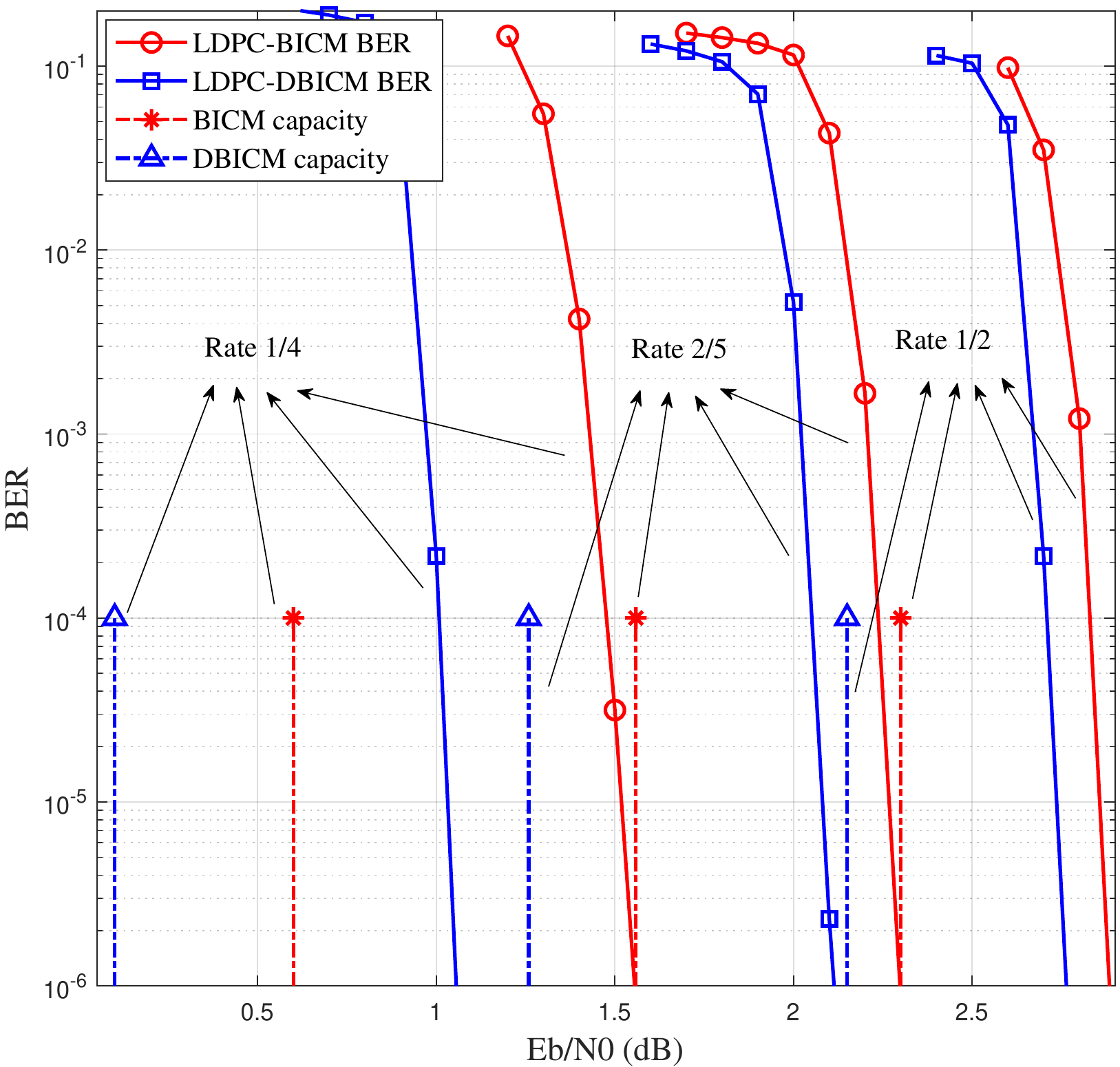}
			\label{fig:16qam_result}
		}
		\subfigure[64-QAM.]
		{
			\includegraphics[width=0.45\textwidth]{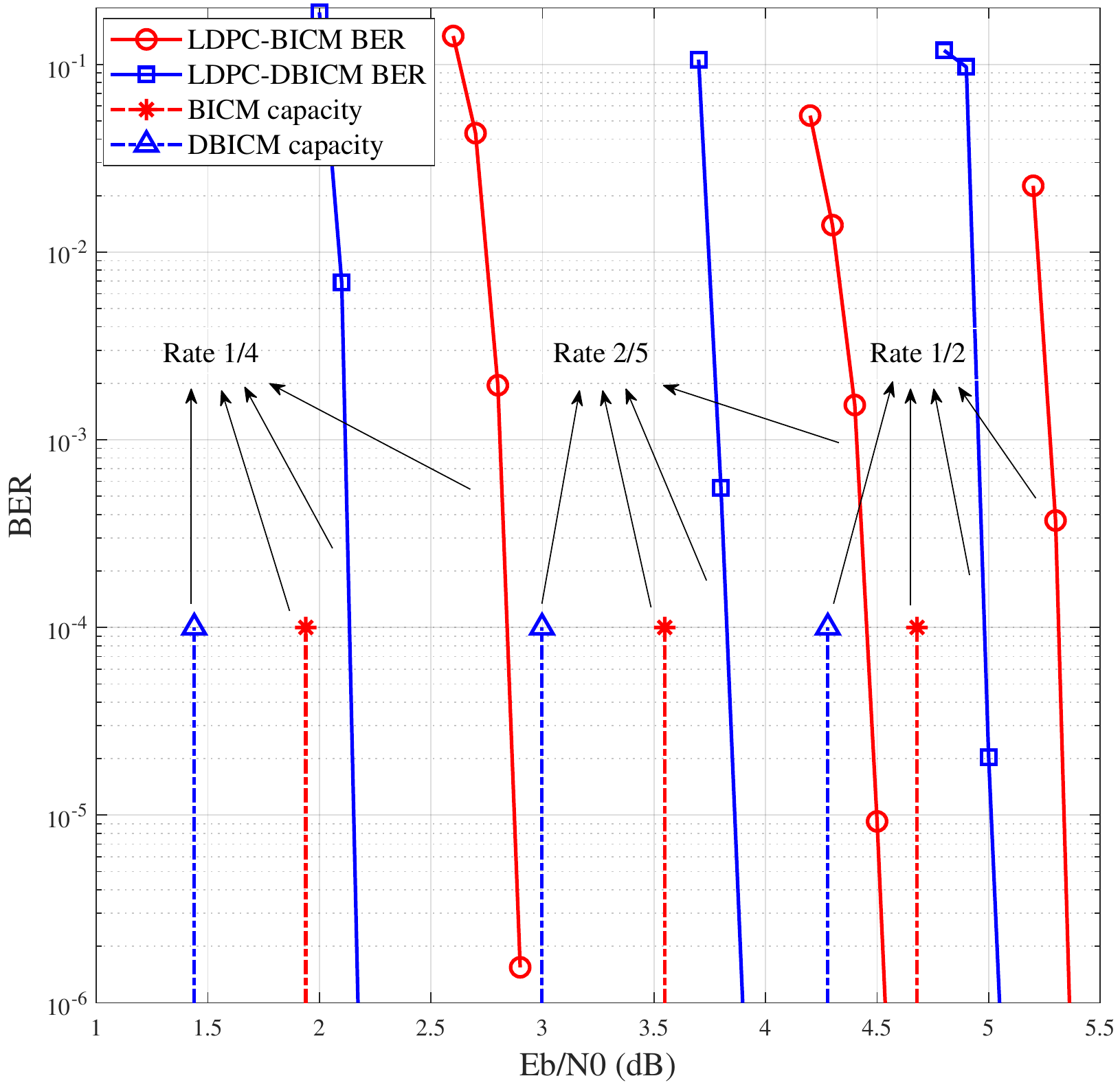}
			\label{fig:64qam_result}
		}
		\caption{BER of LDPC coded 16-QAM and 64-QAM DBICM with the optimal delay scheme.}
		\label{fig:1664}
	\end{figure} 
	
	Figs. \ref{fig:16qam_result} and \ref{fig:64qam_result} show the BER simulation results for LDPC coded $16$-QAM DBICM and $64$-QAM DBICM, respectively. It can be observed that all the designed codes for DBICM outperform the corresponding LDPC coded BICM in the BER performances. For example, the designed LDPC coded DBICM can achieve BER performance gains of $0.5$ dB and $0.7$ dB over LDPC coded BICM for $16$-QAM and $64$-QAM, respectively at rate $1/4$ and BER reaching $10^{-5}$. It can also be observed that the designed LDPC coded DBICM tends to have larger BER performance gains over LDPC coded BICM with the increasing modulation size and decreasing code rate. This result is consistent with the capacity improvement for DBICM over BICM shown in Table \ref{tab:opt_delay_scheme} and the decoding threshold improvements in Table \ref{tab:pvn_16qam} and Table \ref{tab:pvn_64qam}. Finally, the obtained BER performance for all designed LDPC codes is within $1$ dB of their capacity limits and $0.2$ dB of their decoding thresholds at a BER of $10^{-5}$. This shows the effectiveness of our proposed code construction and optimization algorithm for DBICM with LDPC codes in high order modulation systems.

	\begin{table*}[h]
			\caption{optimized $\mathbf{P}$ and the thresholds for uniform Gray labeled $64$-QAM DBICM and BICM with Code Rates $1/4$, $2/5$ and $1/2$.~~~~~~} \label{tab:pvn_64qam}
			\begin{tabular}{|c|c|c|c|c|c|c|c|c|c|c|c|c|}
				\hline
				R   & MOD   & Bit-channels & $d_{v_2}$      & $d_{v_3}$      & $d_{v_4}$      & $d_{v_5}$      & $d_{v_6}$      & $d_{v_7}$    & $d_{v_8}$      & $d_{v_9}$      & $d_{v_{10}}$      & Thresholds  ($E_b/N_0$)    \\ \hline
				\multirow{6}{*}{1/4} & \multirow{3}{*}{DBICM} & \{0, 3\} & 0.2552 &	0.0098 &	0.0006 &	0.0000	 & 0.0021 &	0.0001 &	0.0006 &	0.0043 &	0.0605 & \multirow{3}{*}{$2.1387$}\\ \cline{3-12}
				&                        & \{1, 4\} & 0.2671 &	0.0658 &	0.0004 &	0.0000 &	0.0000 &	0.0000 &	0.0000 &	0.0000 &	0.0000 & \\ \cline{3-12}
				&                        & \{2, 5\} & 0.2701 &	0.0072 &	0.0046 &	0.0051 &	0.0007 &	0.0000 &	0.0003 &	0.0031 &	0.0422 &\\ \cline{2-13}
				& \multirow{3}{*}{BICM} & \{0, 3\} & 0.2449 &	0.0592 &	0.0071 &	0.0037 &	0.0016 &	0.0007 &	0.0002 &	0.0002 &	0.0159 & \multirow{3}{*}{$2.8516$}\\ \cline{3-12}
				&                        & \{1, 4\} & 0.2109 &	0.0264 &	0.0012 &	0.0004 &	0.0038 &	0.0015 &	0.0001 &	0.0002 &	0.0887 &\\ \cline{3-12}
				&                        & \{2, 5\} & 0.3249 &	0.0062 &	0.0013 &	0.0008 &	0.0001 &	0.0000 &	0.0000 &	0.0000 &	0.0000 &\\ \hline
				\multirow{6}{*}{2/5} & \multirow{3}{*}{DBICM} & \{0, 3\} & 0.2058 & 0.0033 & 0.0065 & 0.0162 & 0.0000 & 0.0011 & 0.0054 & 0.0711 & 0.0240 & \multirow{3}{*}{$3.6230$}\\ \cline{3-12} 
				&                        & \{1, 4\} & 0.2567 & 0.0724 & 0.0018 & 0.0020 & 0.0001 & 0.0002 & 0.0001 & 0.0001 & 0.0001 &\\ \cline{3-12} 
				&                        & \{2, 5\} & 0.2107 & 0.1222 & 0.0003 & 0.0001 & 0.0000 & 0.0000 & 0.0000 & 0.0000 & 0.0000 &\\ \cline{2-13} 
				& \multirow{3}{*}{BICM}  & \{0, 3\} & 0.2216 & 0.0059 & 0.0042 & 0.0000 & 0.0001 & 0.0075 & 0.0001 & 0.0923 & 0.0017 & \multirow{3}{*}{$4.1504$}\\ \cline{3-12} 
				&                        & \{1, 4\} & 0.2230 & 0.0964 & 0.0003 & 0.0000 & 0.0002 & 0.0107 & 0.0000 & 0.0008 & 0.0019 &\\ \cline{3-12} 
				&                        & \{2, 5\} & 0.2179 & 0.1151 & 0.0002 & 0.0001 & 0.0000 & 0.0000 & 0.0000 & 0.0000 & 0.0000 & \\ \hline
				\multirow{6}{*}{1/2} & \multirow{3}{*}{DBICM} & \{0, 3\}                      & 0.2401 & 0.0071 & 0.0005 & 0.0010 & 0.0002 & 0.0000 & 0.0000 & 0.0005 & 0.0840 & \multirow{3}{*}{$4.8340$} \\ \cline{3-12} 
				&                        & \{1, 4\}                      & 0.1576 & 0.1398 & 0.0004 & 0.0009 & 0.0000 & 0.0000 & 0.0000 & 0.0004 & 0.0341 &\\ \cline{3-12} 
				&                        & \{2, 5\}                      & 0.1957 & 0.0995 & 0.0017 & 0.0001 & 0.0001 & 0.0000 & 0.0000 & 0.0002 & 0.0359 &\\ \cline{2-13} 
				& \multirow{3}{*}{BICM}  & \{0, 3\}                      & 0.1931 & 0.0652 & 0.0017 & 0.0002 & 0.0000 & 0.0003 & 0.0000 & 0.0000 & 0.0728 & \multirow{3}{*}{$5.2051$}\\ \cline{3-12} 
				&                        & \{1, 4\}                      & 0.2215 & 0.0855 & 0.0005 & 0.0015 & 0.0001 & 0.0001 & 0.0000 & 0.0001 & 0.0240 &\\ \cline{3-12} 
				&                        & \{2, 5\}                      & 0.1789 & 0.0957 & 0.0004 & 0.0002 & 0.0002 & 0.0000 & 0.0000 & 0.0000 & 0.0580 &\\ \hline
			\end{tabular}
		\end{table*}
	\section{Conclusion} \label{sec:con}
	In this paper, we have designed and investigated the performance of uniform Gray labeled $M$-QAM DBICM with LDPC codes. Two important properties regarding the DBICM capacity and its underlying delay schemes are presented. By applying these two properties, we found the optimal delay schemes for uniform Gray labeled $M$-QAM DBICM systems. Moreover, we proposed a code optimization algorithm and a constrained PEG-like algorithm to design capacity-approaching irregular LDPC codes for DBICM systems with the optimal delay schemes. Numerical results show that the designed LDPC coded DBICM is within $1$ dB of their capacity limits and it significantly outperforms their BICM counterparts for various rates and modulations. In future work, constellation design and quasi-cyclic LDPC code design for DBICM and DBICM-ID will be considered to further improve the decoding threshold and for hardware implementation.
	
	\section*{Acknowledgment}
	The authors would like to thank the associate editor and the anonymous reviewers for their constructive and valuable comments on the earlier versions of this paper. Especially, we would like to acknowledging that the proof for capacity-achieving delay schemes has been provided by one of the reviewers.
	\fnbelowfloat
	\appendices
	\section{}\label{app:added_proof}
	We prove that delay schemes $\mathbf{T} = [0,1,\cdots,m-1]$ achieves the constellation constrained capacity  $C_{\text{CM}}$. Consider an input $X$ distributed over a constellation with size $2^m$ whose label is $\bm{B} = \{B_0,\cdots, B_{m-1}\}$ and the corresponding channel output of a memory-less channel is $Y$, the capacity of DBICM with delay scheme $\mathbf{T}$ is 
		\begin{align}
		C_{\text{DBICM}}^{[0,1,\cdots,m-1]} = &\sum_{i=0}^{m-1}I(B_i; Y | B_0 \cdots B_{i-1})=I(\bm{B}; Y) \notag \\
		= & I(X;Y) = C_{\text{CM}}.
		\end{align}
	\section{Proof of Theorem \ref{thm:superposition}} \label{pf:proof_superposition}
	We prove that Eq. (\ref{eq:thm_1_2}) holds. Eq (\ref{eq:thm_1_1}) can be proved in a similar manner. 
	
	Define $\mathbf{T}_1\triangleq[\mathbf{T}_{\mathcal{A}},\mathbf{T}_{\mathcal{B}}]$ and $\mathbf{T}_2\triangleq[\mathbf{T}_{\mathcal{A}}',\mathbf{T}_{\mathcal{B}}]$, where $\mathbf{T}_{\mathcal{A}}\neq\mathbf{T}_{\mathcal{A}}'$. The collection of delayed coded bit labels for delay schemes $\mathbf{T}_1$ and $\mathbf{T}_2$ are $\mathcal{D}_1=\{i|T_{1_i}\neq0\}$ and $\mathcal{D}_2=\{i|T_{2_i}\neq0\}$, respectively. Furthermore, we define the collection of delayed coded bits labels for delay schemes $\mathbf{T}_{\mathcal{A}}$, $\mathbf{T}_{\mathcal{B}}$ and $\mathbf{T}_{\mathcal{A}}'$ as $\mathcal{D_A}\triangleq\mathcal{D}_1\cap\mathcal{A}$, $\mathcal{D_B}\triangleq\mathcal{D}_1\cap\mathcal{B}$ and $\mathcal{D'_A}\triangleq\mathcal{D}_2\cap\mathcal{A}$, respectively. We denote the realization of the delayed coded bits in $\mathcal{D}_1$, $\mathcal{D}_2$, $\mathcal{D_A}$, $\mathcal{D'_A}$, and $\mathcal{D_B}$ by $\mathbf{b}_{\mathcal{D}_1}$, $\mathbf{b}_{\mathcal{D}_1}$, $\mathbf{b}_{\mathcal{D_A}}$, $\mathbf{b}_{\mathcal{D'_A}}$, and $\mathbf{b}_{\mathcal{D_B}}$, distinctively. Obviously, the following relationships hold
		\begin{align}\label{eq:D1D2}
			&\mathcal{D}_1=\mathcal{D_A}\cup\mathcal{D_B}, 
			\quad 
			\mathcal{D}_2=\mathcal{D'_A}\cup\mathcal{D_B}, \notag \\
			&\mathbf{b}_{\mathcal{D}_1}=[{\mathbf{b}_{\mathcal{D_A}},\mathbf{b}_{\mathcal{D_B}}}],
			\quad
			\mathbf{b}_{\mathcal{D}_2}=[{\mathbf{b}_{\mathcal{D'_A}},\mathbf{b}_{\mathcal{D_B}}}].
		\end{align} According to Eq. (\ref{eq:C_dbicm_bit}), the bit-channel capacities of the $k$-th undelayed bit, for $k\in\mathcal{B}\setminus\mathcal{D_B}$, in the uniform Gray labeled $M$-QAM DBICM systems with delay schemes $\mathbf{T}_1$ and $\mathbf{T}_2$ are 
	\begin{align}\label{eq:D1}
 &C_{k,\text{DBICM}}^{\mathbf{T}_1}(M\text{-QAM}) \notag \\
 = & 1 - \frac{1}{2^{\abs{\mathcal{D}_1}}} \mathlarger{\sum_{\mathbf{b}_{\mathcal{D}_1}\in\mathbb{F}_{2}^{\abs{\mathcal{D}_1}}}}\mathop{\mathbb{E}}_{b,\mathbf{y}|\mathcal{D}_1} \left [\!\log_2\dfrac{\sum_{\mathbf{z} \in \chi^{\mathcal{D}_1}_{{\mathbf{b}_{\mathcal{D}_1}}}}p(\mathbf{y}|\mathbf{z})}{\sum_{\mathbf{z} \in \chi^{k,\mathcal{D}_1}_{b,{\mathbf{b}_{\mathcal{D}_1}}}}p(\mathbf{y}|\mathbf{z})}\! \right] \notag \\
\stackrel{\mathclap{\normalfont\mbox{(\ref{eq:D1D2})}}}{=}&\quad 1 - \frac{1}{2^{\abs{\mathcal{D_B}}}} \mathlarger{\sum_{\mathbf{b}_{\mathcal{D_B}}\in\mathbb{F}_{2}^{\abs{\mathcal{D_B}}}}}\frac{1}{2^{\abs{\mathcal{D_A}}}} \notag \\ 
& \quad \mathlarger{\sum_{\mathbf{b}_{\mathcal{D_A}}\in\mathbb{F}_{2}^{\abs{\mathcal{D_A}}}}}  \mathop{\mathbb{E}}_{b,\mathbf{y}|\mathcal{D}_1} \left [\!\log_2\dfrac{\sum_{\mathbf{z} \in \chi^{\mathcal{D}_1}_{{\mathbf{b}_{\mathcal{D}_1}}}}p(\mathbf{y}|\mathbf{z})}{\sum_{\mathbf{z} \in \chi^{k,\mathcal{D}_1}_{b,{\mathbf{b}_{\mathcal{D}_1}}}}p(\mathbf{y}|\mathbf{z})}\! \right].
	\end{align}
	
\begin{align}\label{eq:D2}
&C_{k,\text{DBICM}}^{\mathbf{T}_2} (M\text{-QAM}) \notag \\
= & 1 - \frac{1}{2^{\abs{\mathcal{D}_2}}}\mathlarger{\sum_{\mathbf{b}_{\mathcal{D}_2}\in\mathbb{F}_{2}^{\abs{\mathcal{D}_2}}}}\mathop{\mathbb{E}}_{b,\mathbf{y}|\mathcal{D}_2} \left [\!\log_2\dfrac{\sum_{\mathbf{z} \in \chi^{\mathcal{D}_2}_{{\mathbf{b}_{\mathcal{D}_2}}}}p(\mathbf{y}|\mathbf{z})}{\sum_{\mathbf{z} \in \chi^{k,\mathcal{D}_2}_{b,{\mathbf{b}_{\mathcal{D}_2}}}}p(\mathbf{y}|\mathbf{z})}\! \right] \notag \\
\stackrel{\mathclap{\normalfont\mbox{(\ref{eq:D1D2})}}}{=} &\quad 1 - \frac{1}{2^{\abs{\mathcal{D_B}}}}\mathlarger{\sum_{\mathbf{b}_{\mathcal{D_B}}\in\mathbb{F}_{2}^{\abs{\mathcal{D_B}}}}}\frac{1}{2^{\abs{\mathcal{D'_A}}}}\notag \\
&\quad \mathlarger{\sum_{\mathbf{b}_{\mathcal{D'_A}}\in\mathbb{F}_{2}^{\abs{\mathcal{D'_A}}}}}\mathop{\mathbb{E}}_{b,\mathbf{y}|\mathcal{D}_2} \left [\!\log_2\dfrac{\sum_{\mathbf{z} \in \chi^{\mathcal{D}_2}_{{\mathbf{b}_{\mathcal{D}_2}}}}p(\mathbf{y}|\mathbf{z})}{\sum_{\mathbf{z} \in \chi^{k,\mathcal{D}_2}_{b,{\mathbf{b}_{\mathcal{D}_2}}}}p(\mathbf{y}|\mathbf{z})}\! \right].
\end{align} 
To prove that Eq. (\ref{eq:D1}) and Eq. (\ref{eq:D2}) are identical, we show that the following holds for $k\in \mathcal{B}\setminus\mathcal{D_B}$
\begin{equation} \label{eq:tool_b}
\frac{\sum_{\mathbf{z} \in \chi^{\mathcal{D}_1}_{{{\mathbf{b}_{\mathcal{D}_1}}}}}p(\mathbf{y}|\mathbf{z})}{\sum_{\mathbf{z} \in \chi^{k,\mathcal{D}_1}_{b,\mathbf{b}_{\mathcal{D}_1}}}p(\mathbf{y}|\mathbf{z})} = \frac{\sum_{\mathbf{z} \in \chi^{\mathcal{D}_2}_{{{\mathbf{b}_{\mathcal{D}_2}}}}}p(\mathbf{y}|\mathbf{z})}{\sum_{\mathbf{z} \in \chi^{k,\mathcal{D}_2}_{b,\mathbf{b}_{\mathcal{D}_2}}}p(\mathbf{y}|\mathbf{z})},
\end{equation} for any pair of $\mathcal{D}_1$ and $\mathcal{D}_2$ as long as Eq. (\ref{eq:D1D2}) holds. Let $\tilde{b}$ be the opposite bit value of $b$. We note that
\begin{align} \label{eq:comp_gen}
& \frac{\sum_{\mathbf{z} \in \chi^{\mathcal{D}_1}_{{{\mathbf{b}_{\mathcal{D}_1}}}}}p(\mathbf{y}|\mathbf{z})}{\sum_{\mathbf{z} \in \chi^{k,\mathcal{D}_1}_{b,\mathbf{b}_{\mathcal{D}_1}}}p(\mathbf{y}|\mathbf{z})} - \frac{\sum_{\mathbf{z} \in \chi^{\mathcal{D}_2}_{{{\mathbf{b}_{\mathcal{D}_2}}}}}p(\mathbf{y}|\mathbf{z})}{\sum_{\mathbf{z} \in \chi^{k,\mathcal{D}_2}_{b,\mathbf{b}_{\mathcal{D}_2}}}p(\mathbf{y}|\mathbf{z})} \notag\\
=\quad & \frac{\sum_{\mathbf{z} \in \chi^{\mathcal{D}_1}_{{{\mathbf{b}_{\mathcal{D}_1}}}}}p(\mathbf{y}|\mathbf{z})\sum_{\mathbf{z} \in \chi^{k,\mathcal{D}_2}_{b,\mathbf{b}_{\mathcal{D}_2}}}p(\mathbf{y}|\mathbf{z})}{\sum_{\mathbf{z} \in \chi^{k,\mathcal{D}_1}_{b,\mathbf{b}_{\mathcal{D}_1}}}p(\mathbf{y}|\mathbf{z})\sum_{\mathbf{z} \in \chi^{k,\mathcal{D}_2}_{b,\mathbf{b}_{\mathcal{D}_2}}}p(\mathbf{y}|\mathbf{z})} \notag\\
& -\frac{\sum_{\mathbf{z} \in \chi^{\mathcal{D}_2}_{{{\mathbf{b}_{\mathcal{D}_2}}}}}p(\mathbf{y}|\mathbf{z})\sum_{\mathbf{z} \in \chi^{k,\mathcal{D}_1}_{b,\mathbf{b}_{\mathcal{D}_1}}}p(\mathbf{y}|\mathbf{z})}{\sum_{\mathbf{z} \in \chi^{k,\mathcal{D}_1}_{b,\mathbf{b}_{\mathcal{D}_1}}}p(\mathbf{y}|\mathbf{z})\sum_{\mathbf{z} \in \chi^{k,\mathcal{D}_2}_{b,\mathbf{b}_{\mathcal{D}_2}}}p(\mathbf{y}|\mathbf{z})} \notag\\
=\quad & \frac{\left(\sum_{\mathbf{z} \in \chi^{k,\mathcal{D}_1}_{b,{{\mathbf{b}_{\mathcal{D}_1}}}}}p(\mathbf{y}|\mathbf{z}) + \sum_{\mathbf{z} \in \chi^{k,\mathcal{D}_1}_{\tilde{b},{{\mathbf{b}_{\mathcal{D}_1}}}}}p(\mathbf{y}|\mathbf{z})\right)\sum_{\mathbf{z} \in \chi^{k,\mathcal{D}_2}_{b,\mathbf{b}_{\mathcal{D}_2}}}p(\mathbf{y}|\mathbf{z})}{\sum_{\mathbf{z} \in \chi^{k,\mathcal{D}_1}_{b,\mathbf{b}_{\mathcal{D}_1}}}p(\mathbf{y}|\mathbf{z})\sum_{\mathbf{z} \in \chi^{k,\mathcal{D}_2}_{b,\mathbf{b}_{\mathcal{D}_2}}}p(\mathbf{y}|\mathbf{z})} \notag\\
& -\frac{\left(\sum_{\mathbf{z} \in \chi^{k,\mathcal{D}_2}_{b,{{\mathbf{b}_{\mathcal{D}_2}}}}}p(\mathbf{y}|\mathbf{z}) + \sum_{\mathbf{z} \in \chi^{k,\mathcal{D}_2}_{\tilde{b},{{\mathbf{b}_{\mathcal{D}_2}}}}}p(\mathbf{y}|\mathbf{z})\right)\sum_{\mathbf{z} \in \chi^{k,\mathcal{D}_1}_{b,\mathbf{b}_{\mathcal{D}_1}}}p(\mathbf{y}|\mathbf{z})}{\sum_{\mathbf{z} \in \chi^{k,\mathcal{D}_1}_{b,\mathbf{b}_{\mathcal{D}_1}}}p(\mathbf{y}|\mathbf{z})\sum_{\mathbf{z} \in \chi^{k,\mathcal{D}_2}_{b,\mathbf{b}_{\mathcal{D}_2}}}p(\mathbf{y}|\mathbf{z})} \notag\\
=\quad & \frac{\sum_{\mathbf{z} \in \chi^{k,\mathcal{D}_1}_{\tilde{b},{{\mathbf{b}_{\mathcal{D}_1}}}}}p(\mathbf{y}|\mathbf{z})\sum_{\mathbf{z} \in \chi^{k,\mathcal{D}_2}_{b,\mathbf{b}_{\mathcal{D}_2}}}p(\mathbf{y}|\mathbf{z})}{\sum_{\mathbf{z} \in \chi^{k,\mathcal{D}_1}_{b,\mathbf{b}_{\mathcal{D}_1}}}p(\mathbf{y}|\mathbf{z})\sum_{\mathbf{z} \in \chi^{k,\mathcal{D}_2}_{b,\mathbf{b}_{\mathcal{D}_2}}}p(\mathbf{y}|\mathbf{z})} \notag\\
& - \frac{\sum_{\mathbf{z} \in \chi^{k,\mathcal{D}_2}_{\tilde{b},{{\mathbf{b}_{\mathcal{D}_2}}}}}p(\mathbf{y}|\mathbf{z})\sum_{\mathbf{z} \in \chi^{k,\mathcal{D}_1}_{b,\mathbf{b}_{\mathcal{D}_1}}}p(\mathbf{y}|\mathbf{z})}{\sum_{\mathbf{z} \in \chi^{k,\mathcal{D}_1}_{b,\mathbf{b}_{\mathcal{D}_1}}}p(\mathbf{y}|\mathbf{z})\sum_{\mathbf{z} \in \chi^{k,\mathcal{D}_2}_{b,\mathbf{b}_{\mathcal{D}_2}}}p(\mathbf{y}|\mathbf{z})} \notag\\
=\quad & \frac{\sum_{\mathbf{z} \in \chi^{k,\mathcal{D}_1}_{\tilde{b},{{\mathbf{b}_{\mathcal{D}_1}}}}}e^{-\frac{\norm{\mathbf{y}-\mathbf{z}}^2}{2\sigma^2}}\sum_{\mathbf{z} \in \chi^{k,\mathcal{D}_2}_{b,\mathbf{b}_{\mathcal{D}_2}}}e^{-\frac{\norm{\mathbf{y}-\mathbf{z}}^2}{2\sigma^2}}}{\sum_{\mathbf{z} \in \chi^{k,\mathcal{D}_1}_{b,\mathbf{b}_{\mathcal{D}_1}}}e^{-\frac{\norm{\mathbf{y}-\mathbf{z}}^2}{2\sigma^2}}\sum_{\mathbf{z} \in \chi^{k,\mathcal{D}_2}_{b,\mathbf{b}_{\mathcal{D}_2}}}e^{-\frac{\norm{\mathbf{y}-\mathbf{z}}^2}{2\sigma^2}}} \notag \\
& - \frac{ \sum_{\mathbf{z} \in \chi^{k,\mathcal{D}_2}_{\tilde{b},{{\mathbf{b}_{\mathcal{D}_2}}}}}e^{-\frac{\norm{\mathbf{y}-\mathbf{z}}^2}{2\sigma^2}}\sum_{\mathbf{z} \in \chi^{k,\mathcal{D}_1}_{b,\mathbf{b}_{\mathcal{D}_1}}}e^{-\frac{\norm{\mathbf{y}-\mathbf{z}}^2}{2\sigma^2}}}{\sum_{\mathbf{z} \in \chi^{k,\mathcal{D}_1}_{b,\mathbf{b}_{\mathcal{D}_1}}}e^{-\frac{\norm{\mathbf{y}-\mathbf{z}}^2}{2\sigma^2}}\sum_{\mathbf{z} \in \chi^{k,\mathcal{D}_2}_{b,\mathbf{b}_{\mathcal{D}_2}}}e^{-\frac{\norm{\mathbf{y}-\mathbf{z}}^2}{2\sigma^2}}} \notag \\
= \quad & \frac{\sum_{\mathbf{z}_{0}}\sum_{\mathbf{z}_{1}}e^{-\frac{\norm{\mathbf{y}-\mathbf{z}_{0}}^2+\norm{\mathbf{y}-\mathbf{z}_{1}}^2}{2\sigma^2}} - \sum_{\mathbf{z}_{2}}\sum_{\mathbf{z}_{3}}e^{-\frac{\norm{\mathbf{y}-\mathbf{z}_{2}}^2+\norm{\mathbf{y}-\mathbf{z}_{3}}^2}{2\sigma^2}}}{\sum_{\mathbf{z}_1}e^{-\frac{\norm{\mathbf{y}-\mathbf{z}}^2}{2\sigma^2}}\sum_{\mathbf{z}_3}e^{-\frac{\norm{\mathbf{y}-\mathbf{z}}^2}{2\sigma^2}}},
\end{align} where we define constellation symbols $\mathbf{z}_{0} \in \chi^{k,\mathcal{D}_1}_{\tilde{b},{{\mathbf{b}_{\mathcal{D}_1}}}}$, $\mathbf{z}_{1} \in \chi^{k,\mathcal{D}_2}_{b,\mathbf{b}_{\mathcal{D}_2}}$, $\mathbf{z}_{2} \in \chi^{k,\mathcal{D}_2}_{\tilde{b},{{\mathbf{b}_{\mathcal{D}_2}}}}$, and $\mathbf{z}_{3} \in \chi^{k,\mathcal{D}_1}_{b,{{\mathbf{b}_{\mathcal{D}_1}}}}$. Recall Definition \ref{def:mcbpp}, bits in positions within set $\mathcal{A}=\{0,1,...,\frac{m}{2}-1\}$ and $\mathcal{B}=\{\frac{m}{2},\frac{m}{2}+1,...,m-1\}$ are mapped to the real and imaginary part of the constellation points, respectively. Due to the Gray labeling, constellation symbols with the same real part or imaginary part share the same values for labeled bits in group $\mathcal{A}$ or $\mathcal{B}$. The relationship among the four subsets can be expressed as the following
\begin{equation}\label{eq:complex_presentation_general}
\left\{
\begin{aligned} 
& \Re(\chi^{k,\mathcal{D}_1}_{\tilde{b},{{\mathbf{b}_{\mathcal{D}_1}}}}) = \Re(\chi^{k,\mathcal{D}_1}_{b,{{\mathbf{b}_{\mathcal{D}_1}}}}),\\
& \Re(\chi^{k,\mathcal{D}_2}_{b,\mathbf{b}_{\mathcal{D}_2}}) = \Re(\chi^{k,\mathcal{D}_2}_{\tilde{b},{{\mathbf{b}_{\mathcal{D}_2}}}}),\\
& \Im(\chi^{k,\mathcal{D}_1}_{\tilde{b},{{\mathbf{b}_{\mathcal{D}_1}}}}) = \Im(\chi^{k,\mathcal{D}_2}_{\tilde{b},{{\mathbf{b}_{\mathcal{D}_2}}}}),\\
& \Im(\chi^{k,\mathcal{D}_2}_{b,\mathbf{b}_{\mathcal{D}_2}}) = \Im(\chi^{k,\mathcal{D}_1}_{b,{{\mathbf{b}_{\mathcal{D}_1}}}}).
\end{aligned}
\right.
\end{equation} For example, a $64$-QAM DBICM shown in Fig. \ref{fig:appa_sample_subfigures} with $\mathcal{D}_1=\{2\}$ and $\mathcal{D}_2=\{1\}$, $k=4$ satisfies
	\begin{equation}\label{eq:example1}
	\left\{
	\begin{aligned} 
	& \Re(\chi^{4,\{2\}}_{1,0}) = \Re(\chi^{4,\{2\}}_{0,0}) \text{ and } \Re(\chi^{4,\{2\}}_{1,1}) = \Re(\chi^{4,\{2\}}_{0,1}),\\
	& \Re(\chi^{4,\{1\}}_{0,0}) = \Re(\chi^{4,\{1\}}_{1,0}) \text{ and } \Re(\chi^{4,\{1\}}_{0,1}) = \Re(\chi^{4,\{1\}}_{1,1}),\\
	& \Im(\chi^{4,\{2\}}_{1,0}) = \Re(\chi^{4,\{1\}}_{1,0}) \text{ and } \Im(\chi^{4,\{2\}}_{1,1}) = \Re(\chi^{4,\{1\}}_{1,1}),\\
	& \Im(\chi^{4,\{1\}}_{0,0}) = \Im(\chi^{4,\{2\}}_{1,0}) \text{ and } \Im(\chi^{4,\{1\}}_{0,1}) = \Im(\chi^{4,\{2\}}_{1,1}).
	\end{aligned}
	\right.
	\end{equation}
	\begin{figure}[h]
	\centering
	\subfigure[$k=4,\mathcal{D}_1=\{2\}$.]
	{
		\includegraphics[width=0.45\textwidth]{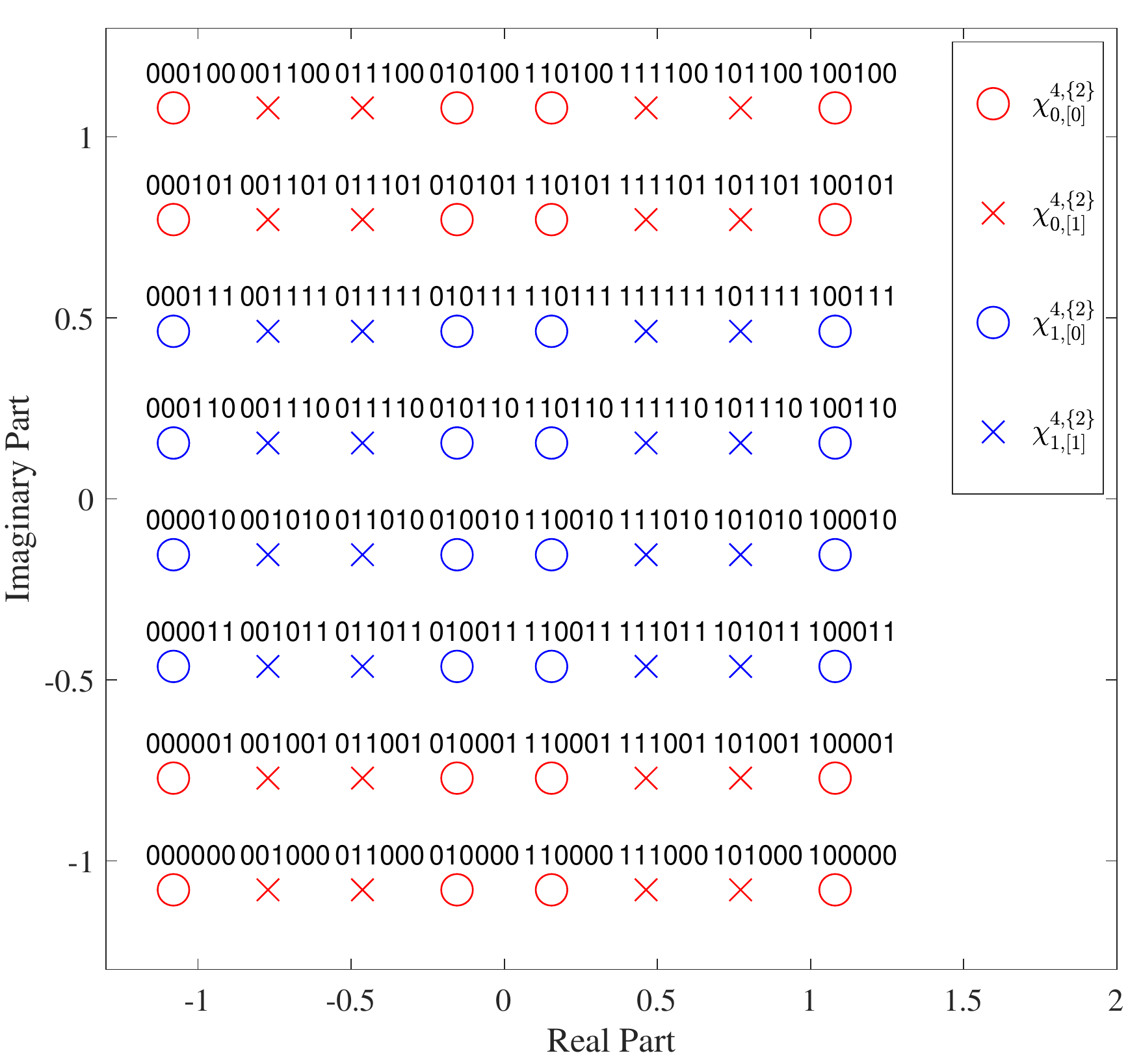}
		\label{fig:appa_gray_64_const_sub_group_a}
	}
	\subfigure[$k=4,\mathcal{D}_2=\{1\}$.]
	{
		\includegraphics[width=0.45\textwidth]{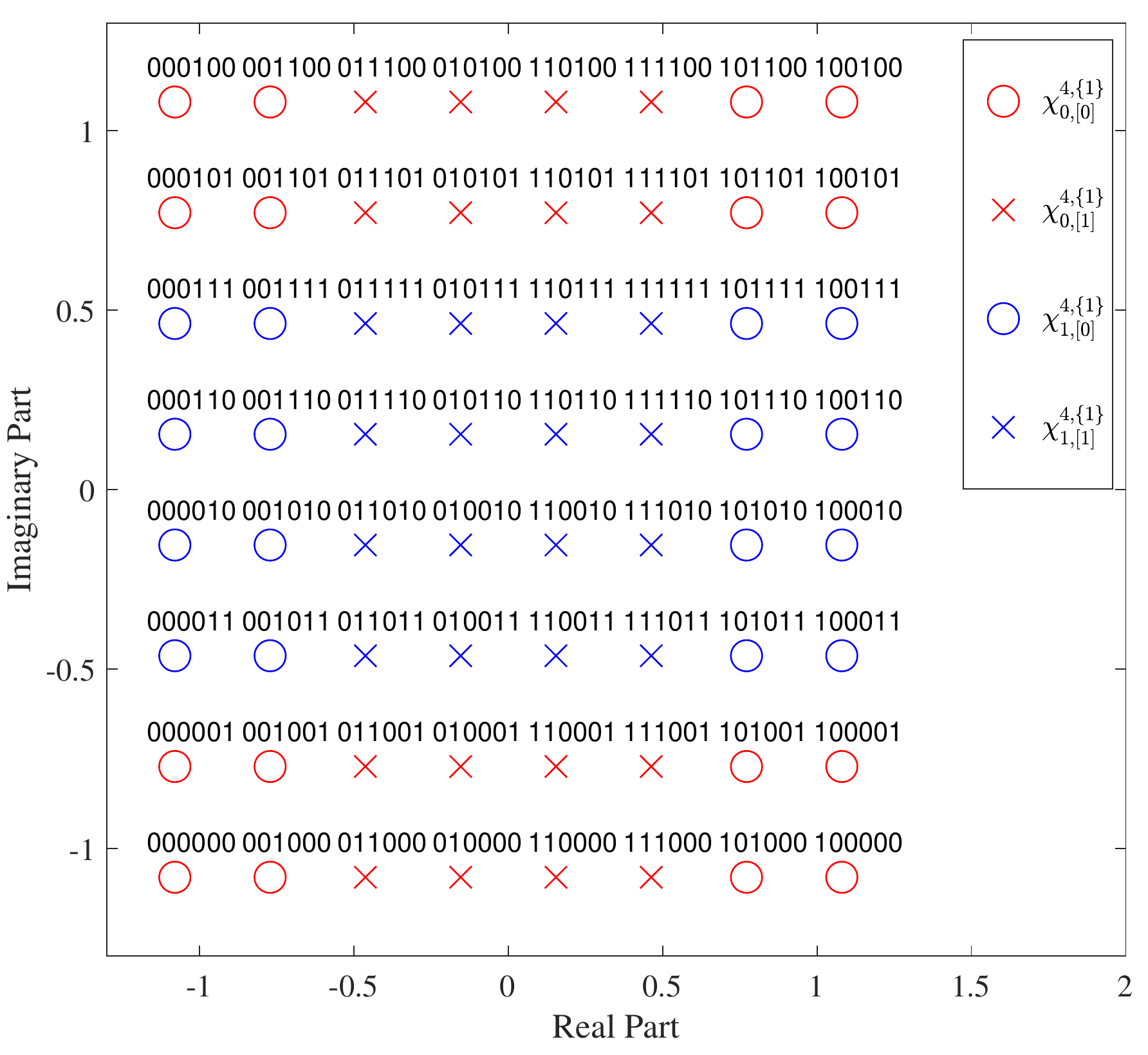}
		\label{fig:appa_gray_64_const_sub_group_b}
	}
	\caption{An example of the geometry relationship among subsets in a uniform Gray labeled $64$-QAM constellation.}
	\label{fig:appa_sample_subfigures}
\end{figure} 

As a result, from Eq. (\ref{eq:comp_gen}),
\begin{align}\label{eq:dis_sqrt_eq_gen}
&\norm{\mathbf{y}-\mathbf{z}_{0}}^2+\norm{\mathbf{y}-\mathbf{z}_{1}}^2\notag\\
=\quad& (\Re(\mathbf{y}) - \Re(\mathbf{z}_{0}))^2 + (\Im(\mathbf{y}) - \Im(\mathbf{z}_{0}))^2  \notag \\
& + (\Re(\mathbf{y}) - \Re(\mathbf{z}_{1}))^2 + (\Im(\mathbf{y}) - \Im(\mathbf{z}_{1}))^2 \notag\\
\stackrel{\mathclap{\normalfont\mbox{(\ref{eq:complex_presentation_general})}}}{=}\quad& (\Re(\mathbf{y}) - \Re(\mathbf{z}_{3}))^2 + (\Im(\mathbf{y}) - \Im(\mathbf{z}_{2}))^2 \notag \\
& + (\Re(\mathbf{y}) - \Re(\mathbf{z}_{2}))^2 + (\Im(\mathbf{y}) - \Im(\mathbf{z}_{3}))^2 \notag\\
=\quad& \norm{\mathbf{y}-\mathbf{z}_{3}}^2+\norm{\mathbf{y}-\mathbf{z}_{2}}^2.
\end{align}
Substituting Eq. (\ref{eq:dis_sqrt_eq_gen}) into Eq. (\ref{eq:comp_gen}), we obtain Eq. (\ref{eq:tool_b}). Hence, Eq. (\ref{eq:D1}) and Eq. (\ref{eq:D2}) are equal. For $k\in\mathcal{B}\setminus\mathcal{D_B}$, this leads to
	\begin{align}\label{eq:appa_mid_step} &\frac{1}{2^{\abs{\mathcal{D_A}}}}\mathlarger{\sum_{\mathbf{b}_{\mathcal{D_A}}\in\mathbb{F}_{2}^{\abs{\mathcal{D_A}}}}}\mathop{\mathbb{E}}_{b,\mathbf{y}|\mathcal{D}_1} \left [\!\log_2\dfrac{\sum_{\mathbf{z} \in \chi^{\mathcal{D}_1}_{{\mathbf{b}_{\mathcal{D}_1}}}}p(\mathbf{y}|\mathbf{z})}{\sum_{\mathbf{z} \in \chi^{k,\mathcal{D}_1}_{b,{\mathbf{b}_{\mathcal{D}_1}}}}p(\mathbf{y}|\mathbf{z})}\! \right] \notag \\
		  \stackrel{\mathclap{\normalfont\mbox{(a)}}}{=} & \mathop{\mathbb{E}}_{b,\mathbf{y}|\mathcal{D_B}} \left [\!\log_2\dfrac{\sum_{\mathbf{z} \in \chi^k_{\mathbf{b}_{\mathcal{D_B}}}}p(\mathbf{y}|\mathbf{z})}{\sum_{\mathbf{z} \in \chi^{k,\mathcal{D_B}}_{b,\mathbf{b}_{\mathcal{D_B}}}}p(\mathbf{y}|\mathbf{z})}\! \right],
	\end{align} where (a) follows by letting ${\mathcal{D_A}}=\emptyset$. Consequently, we can write Eq. (\ref{eq:D1}) into 
	\begin{align}\label{eq:newD1}
&\qquad C_{k,\text{DBICM}}^{\mathbf{T}_1} (M\text{-QAM}) \notag \\
&\stackrel{\mathclap{\normalfont\mbox{(\ref{eq:appa_mid_step})}}}{=}\quad  1 - \frac{1}{2^{\abs{\mathcal{D_B}}}}\mathlarger{\sum_{\mathbf{b}_{\mathcal{D_B}}\in\mathbb{F}_{2}^{\abs{\mathcal{D_B}}}}}\mathop{\mathbb{E}}_{b,\mathbf{y}|\mathcal{D_B}} \left [\!\log_2\dfrac{\sum_{\mathbf{z} \in \chi^{\mathcal{D_B}}_{{\mathbf{b}_{\mathcal{D_B}}}}}p(\mathbf{y}|\mathbf{z})}{\sum_{\mathbf{z} \in \chi^{k,\mathcal{D_B}}_{b,{\mathbf{b}_{\mathcal{D_B}}}}}p(\mathbf{y}|\mathbf{z})}\! \right] \notag \\
&=  1 - \frac{1}{2^{\abs{\mathcal{D_B}}}}\mathlarger{\sum_{\mathbf{b}_{\mathcal{D_B}}\in\mathbb{F}_{2}^{\abs{\mathcal{D_B}}}}}\mathop{\mathbb{E}}_{b,\mathbf{y}|\mathcal{D_B}} \notag \\
& \qquad \left [\!\log_2\dfrac{\sum_{\mathbf{z} \in \chi^{\mathcal{D_B}}_{{\mathbf{b}_{\mathcal{D_B}}}}}p(\Re(\mathbf{y})|\Re(\mathbf{z})) p(\Im(\mathbf{y})|\Im(\mathbf{z}))}{\sum_{\mathbf{z} \in \chi^{k,\mathcal{D_B}}_{b,{\mathbf{b}_{\mathcal{D_B}}}}}p(\Re(\mathbf{y})|\Re(\mathbf{z})) p(\Im(\mathbf{y})|\Im(\mathbf{z}))}\! \right] \notag \\
&\stackrel{\mathclap{\normalfont\mbox{(b)}}}{=}   1 - \frac{1}{2^{\abs{\mathcal{D_B}}}}\mathlarger{\sum_{\mathbf{b}_{\mathcal{D_B}}\in\mathbb{F}_{2}^{\abs{\mathcal{D_B}}}}}\mathop{\mathbb{E}}_{b,\mathbf{y}|\mathcal{D_B}} \notag \\
& \quad \left [\!\log_2\dfrac{\left(\sum_{\mathbf{z} \in \Re(\chi^{\mathcal{D_B}}_{{\mathbf{b}_{\mathcal{D_B}}}})}p(\Re(\mathbf{y})|\mathbf{z})\right)\left(\sum_{\mathbf{z} \in \Im(\chi^{\mathcal{D_B}}_{{\mathbf{b}_{\mathcal{D_B}}}})} p(\Im(\mathbf{y})|\mathbf{z})\right)}{\left(\sum_{\mathbf{z} \in \Re(\chi^{k,\mathcal{D_B}}_{b,{\mathbf{b}_{\mathcal{D_B}}}})}p(\Re(\mathbf{y})|\mathbf{z})\right) \left(\sum_{\mathbf{z} \in \Im(\chi^{k,\mathcal{D_B}}_{b,{\mathbf{b}_{\mathcal{D_B}}}})} p(\Im(\mathbf{y})|\mathbf{z})\right)}\! \right] \notag \\
&\stackrel{\mathclap{\normalfont\mbox{(c)}}}{=}  1 - \frac{1}{2^{\abs{\mathcal{D_B}}}}\mathlarger{\sum_{\mathbf{b}_{\mathcal{D_B}}\in\mathbb{F}_{2}^{\abs{\mathcal{D_B}}}}}\mathop{\mathbb{E}}_{b,\mathbf{y}|\mathcal{D_B}} \left [\!\log_2\dfrac{\sum_{\mathbf{z} \in \Im(\chi^{\mathcal{D_B}}_{{\mathbf{b}_{\mathcal{D_B}}}})}p(\Im(\mathbf{y})|\mathbf{z})}{\sum_{\mathbf{z} \in \Im(\chi^{k,\mathcal{D_B}}_{b,{\mathbf{b}_{\mathcal{D_B}}}})}p(\Im(\mathbf{y})|\mathbf{z})}\! \right] \notag \\
& =C_{k-\frac{m}{2},\text{DBICM}}^{\mathbf{T}_{\mathcal{B}}} (\sqrt{M}\text{-PAM}),
\end{align} where step (b) follows that $\chi^{\mathcal{D_B}}_{\mathbf{b}_{\mathcal{D_B}}}$ and $\chi^{k,\mathcal{D_B}}_{b,\mathbf{b}_{\mathcal{D_B}}}$ can be reconstructed via the Cartesian product of their real and imaginary parts. Step (c) follows $\Re(\chi)=\Re(\chi^{\mathcal{D_B}}_{\mathbf{b}_{\mathcal{D_B}}}) = \Re(\chi^{k,\mathcal{D_B}}_{b,\mathbf{b}_{\mathcal{D_B}}})$ as only labeled bits in group $\mathcal{A}$ is associated with the real part of the constellation symbol. Similarly, substituting Eq. (\ref{eq:appa_mid_step}) into Eq. (\ref{eq:D2}) results in
\begin{equation}\label{eq:newD2}
C_{k,\text{DBICM}}^{\mathbf{T}_2}  (M\text{-QAM}) =C_{k-\frac{m}{2},\text{DBICM}}^{\mathbf{T}_{\mathcal{B}}} (\sqrt{M}\text{-PAM}).
\end{equation} This completes the proof.
	\section{Proof of Theorem \ref{thm:D_D'_symmetric}} \label{pf:proof_D_D'_symmetric}
	\begin{proof} 
		Since $\mathbf{T}$ and $\mathbf{T}'$ are a pair of symmetric scheme according to Def. \ref{def:symmetric_schemes}, i.e., $\mathbf{T}=[\mathbf{T}_{\mathcal{A}}, \mathbf{T}_{\mathcal{B}}]$, $\mathbf{T}'=[\mathbf{T}_{\mathcal{B}}, \mathbf{T}_{\mathcal{A}}]$. We denote the collection of the delayed coded bits in $\mathbf{T}$ and $\mathbf{T}'$ by $\mathcal{D}=\{i|T_i\neq0\}$ and $\mathcal{D'}=\{i|T_{i}'\neq0\}$. The relationship between $k\in\mathcal{D}$ and $k'\in\mathcal{D}'$ satisfies
		\begin{equation}\label{eq:kk'DD'}
			\begin{array}{ll}
			k'=k+\frac{m}{2}, k \in \mathcal{D}\cap\mathcal{A}, k' \in \mathcal{D}'\cap\mathcal{B}, \\
			k'=k-\frac{m}{2}, k \in \mathcal{D}\cap\mathcal{B}, k' \in \mathcal{D}'\cap\mathcal{A}.
			\end{array}
		\end{equation} We also denote the collection of the undelayed coded bits in $\mathbf{T}$ and $\mathbf{T}'$ by $\tilde{\mathcal{D}}=\{i|T_i=0\}$ and $\tilde{\mathcal{D'}}=\{i|T_{i}'=0\}$. Similarly, the relationship between $k\in\tilde{\mathcal{D}}$ and $k'\in\tilde{\mathcal{D}'}$ satisfies
		\begin{equation}\label{eq:k'_D'}
		\begin{array}{ll}
		k'=k+\frac{m}{2}, k \in \tilde{\mathcal{D}}\cap\mathcal{A}, k' \in \tilde{\mathcal{D}'}\cap\mathcal{B}, \\
		k'=k-\frac{m}{2}, k \in \tilde{\mathcal{D}}\cap\mathcal{B}, k' \in \tilde{\mathcal{D}'}\cap\mathcal{A}.
		\end{array}
		\end{equation} 
		Following Eq. (\ref{eq:C_dbicm_overall}), the delayed coded bit in a DBICM system has the same bit-channel capacity as that bit in a BICM system, we have
		\begin{align} \label{eq:C_dbicm_T}
			&C_{\text{DBICM}}^{\mathbf{T}}(\text{$M$-QAM}) \notag \\
			= & \sum\limits_{k\in\mathcal{D}}C_{k,\text{BICM}}(\text{$M$-QAM})  + \sum\limits_{k\in\tilde{\mathcal{D}}} C_{k,\text{DBICM}}^{\mathbf{T}}(\text{$M$-QAM}).
		\end{align} 

\begin{align}\label{eq:C_dbicm_T'}
&C_{\text{DBICM}}^{\mathbf{T'}}(\text{$M$-QAM}) \notag \\
= & \sum\limits_{k'\in\mathcal{D'}}C_{k',\text{BICM}}(\text{$M$-QAM})  + \sum\limits_{k'\in\tilde{\mathcal{D'}}} C_{k',\text{DBICM}}^{\mathbf{T'}}(\text{$M$-QAM}).
\end{align}
Using Definition \ref{def:mcbpp} and Eq. (\ref{eq:kk'DD'}), each pair of symmetric bits share identical bit channel capacity, we have
\begin{align}\label{eq:sym_apb}
	\sum\limits_{k\in\mathcal{D}}C_{k,\text{BICM}}(\text{$M$-QAM}) = \sum\limits_{k'\in\mathcal{D}'}C_{k',\text{BICM}}(\text{$M$-QAM}).
\end{align}
By using Eqs. (\ref{eq:thm_1_1}-\ref{eq:thm_1_2}) from Theorem \ref{thm:superposition}, the second terms in Eqs. (\ref{eq:C_dbicm_T}-\ref{eq:C_dbicm_T'}) can be written as
\begin{align}
\sum\limits_{k\in\tilde{\mathcal{D}}} C_{k,\text{DBICM}}^{\mathbf{T}}(\text{$M$-QAM}) = & \sum\limits_{k\in\mathcal{B}\cap\tilde{\mathcal{D}}} C_{k-\frac{m}{2},\text{DBICM}}^{\mathbf{T}_\mathcal{B}}(\text{$\sqrt{M}$-QAM}) \notag \\
 &+ \sum\limits_{k\in\mathcal{A}\cap\tilde{\mathcal{D}}} C_{k,\text{DBICM}}^{\mathbf{T}_\mathcal{A}}(\text{$\sqrt{M}$-QAM}),
\end{align}
\begin{align}
\sum\limits_{k'\in\tilde{\mathcal{D}'}} C_{k',\text{DBICM}}^{\mathbf{T}'}(\text{$M$-QAM}) = & \sum\limits_{k'\in\mathcal{B}\cap\tilde{\mathcal{D}'}} C_{k'-\frac{m}{2},\text{DBICM}}^{\mathbf{T}_\mathcal{A}}(\text{$\sqrt{M}$-QAM}) \notag \\ &+ \sum\limits_{k'\in\mathcal{A}\cap\tilde{\mathcal{D}'}} C_{k',\text{DBICM}}^{\mathbf{T}_\mathcal{B}}(\text{$\sqrt{M}$-QAM}).
\end{align} From Eq. (\ref{eq:k'_D'}), we note that
\begin{align}\label{eq:apb_1}
&\sum\limits_{k'\in\mathcal{A}\cap\tilde{\mathcal{D}'}} C_{k',\text{DBICM}}^{\mathbf{T}_\mathcal{B}}(\text{$\sqrt{M}$-QAM}) \notag \\ 
= & \sum\limits_{k\in\mathcal{B}\cap\tilde{\mathcal{D}}} C_{k-\frac{m}{2},\text{DBICM}}^{\mathbf{T}_\mathcal{B}}(\text{$\sqrt{M}$-QAM}),
\end{align}
\begin{align}\label{eq:apb_2}
&\sum\limits_{k'\in\mathcal{B}\cap\tilde{\mathcal{D}'}} C_{k'-\frac{m}{2},\text{DBICM}}^{\mathbf{T}_\mathcal{A}}(\text{$\sqrt{M}$-QAM}) \notag \\
= & \sum\limits_{k\in\mathcal{A}\cap\tilde{\mathcal{D}}} C_{k,\text{DBICM}}^{\mathbf{T}_\mathcal{A}}(\text{$\sqrt{M}$-QAM}).
\end{align}
Substituting Eqs.(\ref{eq:sym_apb}-\ref{eq:apb_2}) into Eqs.(\ref{eq:C_dbicm_T}-\ref{eq:C_dbicm_T'}) gives
\begin{equation}
	C^{\mathbf{T}}_{\text{DBICM}}(M\text{-QAM}) = C^{\mathbf{T}'}_{\text{DBICM}}(M\text{-QAM})
\end{equation} This completes the proof.
	\end{proof}

	\section{Proof of Proposition \ref{prop:p}}\label{pf:proof_prop_1}
	We start from Eq. (\ref{eq:C_dbicm_overall}):
	\begin{align}
	&\quad C^{\mathbf{T}}_{\text{DBICM}}(\text{$M$-QAM})\notag \\
	= &\quad \sum\limits_{j\in\mathcal{D}}C_{j,\text{BICM}}(\text{$M$-QAM}) + \sum\limits_{k\in\tilde{\mathcal{D}}}C_{k,\text{DBICM}}^{\mathbf{T}}(\text{$M$-QAM}) \notag\\ 
	\stackrel{\mathclap{\normalfont\mbox{(d)}}}{=}&\quad \sum\limits_{j\in\mathcal{D}}C_{j,\text{BICM}}(\text{$M$-QAM}) +  \sum\limits_{k\in\tilde{\mathcal{D}}\cap\mathcal{A}}C_{k,\text{DBICM}}^{[\mathbf{T}_{\mathcal{A}}, \mathbf{T}_{\mathcal{B}}]}\text{($M$-QAM)} \notag \\
	 & + \quad  \sum\limits_{(k+\frac{m}{2})\in\tilde{\mathcal{D}}\cap\mathcal{B}}C_{k,\text{DBICM}}^{[\mathbf{T}_{\mathcal{B}}, \mathbf{T}_{\mathcal{A}}]}\text{($M$-QAM)} \notag\\ 
	\stackrel{\mathclap{\normalfont\mbox{(e)}}}{=}&\quad \sum\limits_{j\in\mathcal{D}}C_{j,\text{BICM}}(\text{$M$-QAM}) + \sum\limits_{k\in\tilde{\mathcal{D}}\cap\mathcal{A}}C_{k,\text{DBICM}}^{\mathbf{T}_{\mathcal{A}}}\text{($\sqrt{M}$-PAM)} \notag \\
	 & + \quad \sum\limits_{(k+\frac{m}{2})\in\tilde{\mathcal{D}}\cap\mathcal{B}}C_{k,\text{DBICM}}^{\mathbf{T}_{\mathcal{B}}}\text{($\sqrt{M}$-PAM)} \notag\\ 
	\stackrel{\mathclap{\normalfont\mbox{(f)}}}{=}&\quad \sum\limits_{j\in\mathcal{D}\cap\mathcal{A}}C_{j,\text{BICM}}\text{($\sqrt{M}$-PAM)} \notag \\
	& + \quad \sum\limits_{(j+\frac{m}{2})\in\mathcal{D}\cap\mathcal{B}}C_{j,\text{BICM}}\text{($\sqrt{M}$-PAM)} \notag \\
	 & + \quad \sum\limits_{k\in\tilde{\mathcal{D}}\cap\mathcal{A}}C_{k,\text{DBICM}}^{\mathbf{T}_{\mathcal{A}}}\text{($\sqrt{M}$-PAM)} \notag \\
	 & + \quad \sum\limits_{(k+\frac{m}{2})\in\tilde{\mathcal{D}}\cap\mathcal{B}}C_{k,\text{DBICM}}^{\mathbf{T}_{\mathcal{B}}}\text{($\sqrt{M}$-PAM)} \notag \\
	 \quad\stackrel{\mathclap{\normalfont\mbox{(g)}}}{=} & \quad C_{\text{DBICM}}^{\mathbf{T}_{\mathcal{A}}}\text{($\sqrt{M}$-PAM)} + C_{\text{DBICM}}^{\mathbf{T}_{\mathcal{B}}}\text{($\sqrt{M}$-PAM)},
	\end{align} where (d) is by applying Theorem \ref{thm:D_D'_symmetric} to the last term, (e) results from applying Theorem \ref{thm:superposition} to the last two terms, (f) follows that the capacity of BICM $M$-QAM is the sum of the capacities of two BICM $\sqrt{M}$-PAM \cite{1021039}, (g) follows from Eq. (\ref{eq:C_dbicm_overall}) again.

	\bibliographystyle{ieeetr}
	\bibliography{pubs}

\end{document}